\pdfoutput=1
\documentclass[11pt]{article}

\usepackage{filecontents}
\usepackage{titling}
\usepackage[utf8]{inputenc}
\usepackage[english]{babel}
\usepackage{graphicx}
\usepackage{subcaption}
\usepackage{amsmath, amsfonts, amssymb, amsthm, bm, graphicx, mathtools, enumerate,multirow}
\usepackage[affil-it]{authblk}
\usepackage[round]{natbib}
\usepackage{bbm}
\usepackage{booktabs}
\usepackage{enumitem}
\usepackage{fixltx2e}
\usepackage{setspace}
\usepackage{xcolor}
\usepackage[CJKbookmarks=true,
bookmarksnumbered=true,
bookmarksopen=true,
colorlinks=true,
citecolor=blue,
linkcolor=blue,
anchorcolor=blue,
urlcolor=blue]{hyperref}
\usepackage[letterpaper, left=1.2truein, right=1.2truein, top = 1.2truein,bottom = 1.2truein]{geometry}
\usepackage[ruled, vlined, lined, commentsnumbered,linesnumbered]{algorithm2e}
\usepackage{prettyref,soul}

\setlist{nolistsep}
\allowdisplaybreaks

\usepackage{xr}
\newcommand{\supp}[1]{{Section #1 in the supplementary material}}

\DeclareMathOperator*{\argmin}{argmin}

\newtheorem{lemma}{Lemma}

\newtheorem{thm}{Theorem}

\newtheorem{assumption}{Assumption}
\theoremstyle{definition}
\newtheorem{remark}{Remark}
\usepackage{float}
\floatstyle{plaintop}
\restylefloat{table}

\def\pr{\mathbb{P}}
\def\E{\mathbb{E}} 
 
\def\Var{\mathrm{Var}}

\newcommand{\tp}{\intercal}

\newcommand{\bigO}{\ensuremath{\mathop{}\mathopen{}\mathcal{O}\mathopen{}}}
\newcommand{\smallO}{ \scalebox{0.7}{$\mathcal{O}$}}
\newcommand{\bigOp}{\bigO_\mathrm{p}}
\newcommand{\smallOp}{\smallO_\mathrm{p}}
\newcommand{\ind}{\mbox{$\perp\!\!\!\perp$}}
\newcommand{\Hscr}{{\mathcal{H}}}

\newcommand{\Kscr}{{\mathcal{K}}}

\newcommand{\Fscr}{{\mathcal{F}}}

\newcommand{\kernel}{{\kappa}}
\newcommand{\logit}{{\mathrm{logit}}}

\newcommand{\Ak}{A_1}

\newcommand{\nuk}{{\nu_1}}

\newcommand{\tK}{\tilde{K}_{h}}

\newcommand{\ate}{\textsc{ATE\textsubscript{RKHS}}}
\newcommand{\cate}{\textsc{Proposed}}
\newcommand{\ipw}{\textsc{IPW}}

\newcommand{\de}{\mathrm{d}}

\newcommand{\Noweights}{\textsc{REG}}

\newcounter{ucnt}
\newcommand{\newu}{
	\refstepcounter{ucnt}
	\ensuremath{C_{\theucnt}}
}
\newcommand{\oldu}[1]{\ensuremath{C_{\ref*{#1}}}}

\newcounter{lcnt}

\makeatletter

\newcommand{\ltxlabel}[1]{\ltx@label{#1}}

\makeatother

\begin{document}

\def\spacingset#1{\renewcommand{\baselinestretch}%
{#1}\small\normalsize}


  \title{\bf Estimation of Partially Conditional Average Treatment Effect by Hybrid Kernel-covariate Balancing}

\author[1]{Jiayi Wang}
\author[1]{Raymond K. W. Wong}
\author[2]{Shu Yang}
\author[3]{Kwun Chuen Gary Chan}
\affil[1]{Department of Statistics, Texas A\&M University}
\affil[2]{Department of Statistics, North Carolina State University}
\affil[3]{Department of Biostatistics, University of Washington}

\date{}
\maketitle

\bigskip
\begin{abstract}
  We study nonparametric estimation for the partially conditional average treatment effect, defined as the treatment effect function over an interested subset of confounders. 
  We propose a hybrid kernel weighting estimator where the 
  weights
  aim to control
  the balancing error of
  any function of the confounders
  from a reproducing kernel Hilbert space after kernel smoothing over the subset of interested variables. In addition, we present an augmented version of our estimator which can incorporate estimations of outcome mean functions. Based on the representer theorem, gradient-based algorithms can be  applied for solving the corresponding infinite-dimensional optimization problem. Asymptotic properties are studied 
  without any smoothness assumptions for propensity score function or the need of data splitting, relaxing certain existing stringent assumptions. 
  The  numerical performance of the proposed estimator is demonstrated by a simulation study and an application to  the effect of a mother's smoking on a baby's birth weight conditioned on the mother's age.
\end{abstract}

\noindent%

\spacingset{1}

\section{Introduction}\label{sec:intro}
Causal inference often concerns not only the average effect of the treatment
on the outcome but also the conditional average treatment effect (CATE) given
a set of individual characteristics, when treatment effect heterogeneity is expected or of interest.  
Specifically, let $T\in\{0,1\}$ be the treatment assignment, $0$ for control
and $1$ for active treatment, $X\in\mathcal{X}\subset\mathbb{R}^{d}$
 a vector of all pre-treatment confounders, and $Y$  the outcome of interest.
Following the potential outcomes framework, let $Y(t)$ be the potential
outcome, possibly contrary to fact, had the unit received treatment
$t\in\{0,1\}$. Then, the individual treatment effect is $Y(1)-Y(0)$,
and the (fully) CATE can be characterized through $\gamma(x)=\E\{Y(1)-Y(0)\mid X=x\}$, $x\in\mathcal{X}$.
Due to the fundamental problem in causal inference that the potential
outcomes are not jointly observable, identification and estimation
of the CATE in observational studies require further assumptions.
A common assumption is the no unmeasured confounding (UNC) assumption,
requiring $X$ to capture all confounding variables that affect the
treatment assignment and outcome. This often results in a multidimensional
$X$. Given the UNC assumption, many methods have been proposed to
estimate $\gamma(x)$ \citep{nie2017quasi,wager2018estimation,kennedy2020optimal}.
However, in clinical settings, researchers may only concern the variation
of treatment effect over the change of a small subset of covariates
$V\in\mathcal{V}\subseteq\mathcal{X}$, not necessarily the full set
$X$. For example, researchers are interested in estimating the CATE
of smoking (treatment) on birth weight (outcome) given mother's age
but not mother’s educational attainment, although this variable can
be a confounder. In this article, we focus on estimating $\tau(v)=\E\{\gamma(X)\mid V=v\}$ for $v\in\mathcal{V}$, which we refer to as the partially conditional average treatment effect (PCATE).
When $V$ is taken to be $X$, $\tau(v)$ becomes the fully conditional average treatment effect (FCATE)
$\gamma(x)$.
Despite our major focus on cases when $\mathcal{V}$ is a proper subset of $\mathcal{X}$,
the proposed method in this paper does not exclude the setting with $\mathcal{V}=\mathcal{X}$,
which results in the FCATE.

When $V$ contains discrete covariates, one can
divide the whole sample into different groups by constricting the
same values of discrete covariates of $V$ in the same group. Then, as long as there are enough samples in such stratum, $\tau(v)$ can be obtained
by estimating the PCATE over the remaining continuous covariates in $V$
separately for every stratum.
Therefore, for simplicity, we focus on the setups with continuous $V$ \citep{abrevaya2015estimating,lee2017doubly,fan2020estimation,zimmert2019nonparametric, semenova2020debiased} while keeping in mind that
the proposed method can be used to handle $V$ that consists of continuous and discrete variables.
The typical estimation strategy involves two steps.
The first step is to estimate nuisance
parameters including the propensity score function and the outcome
mean functions for the construction of adjusted responses (through weighting and augmentation) that are (asymptotically) unbiased for $\gamma(x)$ given $X=x$. 
The nuisance
parameters  can be estimated by parametric, nonparametric, or even
machine learning models.
This step serves to adjust for confounding biases.
In the second step, 
existing methods typically adopt  nonparametric regression over $V$ using the adjusted responses obtained from the first step. 
However,
these methods suffer from many drawbacks. Firstly, all parametric methods
are potentially sensitive to
model misspecification especially when 
the CATE is complex. On the other hand,
although nonparametric and machine learning methods are flexible,
the first-step  estimator of $\gamma(X)$ with high-dimensional $X$ 
requires stringent assumptions for the possibly low-dimensional PCATE estimation to achieve the optimal convergence rate.  
For example, \cite{abrevaya2015estimating}, \cite{zimmert2019nonparametric}, \cite{fan2020estimation} and \cite{semenova2020debiased} specify restrictive requirements
for the convergence rate of the estimators of the nuisance parameters.
Detailed discussions are provided in Remarks
\ref{rmk:ipw} and \ref{rmk:aipw2}.

Instead of separating confounding adjustment and kernel smoothing in two steps, we propose a new
framework that unifies the confounding adjustment and kernel smoothing
in the weighting step. In particular, we generalize
the idea of covariate balancing weighting in the average treatment
effect (ATE) estimation literature \citep{qin2007empirical,hainmueller2012entropy,imai2014covariate,zubizarreta2015stable,wong2018kernel}.
This generalization, however, is non-trivial because we require covariate balancing
in terms of flexible outcome models between the two treatment groups
given all possible values of $v$. We assume that the outcome models
lie in the reproducing kernel Hilbert space (RKHS, \citealp{wahba1990spline}),
a fairly flexible class of functions of $X$. We then propose covariate
function balancing (CFB) weights that are capable of controlling the balancing
error with respect to the $L_{2}$-norm of any function with a bounded
norm over the RKHS after kernel smoothing.
The construction of the proposed weights specifically involves two
kernels --- the reproducing kernel of the RKHS and the kernel of the kernel smoothing --- and the goal of these weights can be understood as to balance covariate functions generated by the hybrid of these two kernels.
Our method does not
require any smoothness assumptions on the propensity score
model, in sharp contrast to existing methods, and only require mild smoothness assumptions for the outcome
models. Invoking the well-known representer theorem, a finite-dimensional
representation form of optimization objective can be derived and it
can be solved by a gradient-based algorithm. Asymptotic properties
of the proposed estimator are derived under the complex dependency
structure of weights and kernel smoothing. In addition, our proposed
weighting estimator can be slightly modified to incorporate the estimation
of the outcome mean functions, similar to the augmented inverse probability weighting (AIPW) estimator. We show that
the augmentation of the outcome models relaxes the selection of tuning
parameters theoretically.

The rest of paper is organized as follows. Section \ref{sec:setup} provides the basic setup for the CATE estimation. Section \ref{sec: estimator} introduces our proposed CFB weighting estimator, together with the computation techniques. Section \ref{sec:aug_estimator} introduces an augmented version of our proposed estimator. In Section \ref{sec:theory}, the asymptotic properties of the proposed estimators are developed.  A simulation study and a real data application are presented in Sections \ref{sec:sim} and \ref{sec:real}, respectively.
\section{Basic setup\label{sec:setup}}

Suppose $\{(T_{i},Y_{i}(1),Y_{i}(0),X_{i}):i=1,\ldots,N\}$ are $N$ independent
and identically distributed copies of $\{T,Y(1),Y(0),X\}$. We assume that the observed outcome is $Y_{i}=T_{i}Y_{i}(1)+(1-T_{i})Y_{i}(0)$
for $i=1,\ldots,N$. Thus, the observed data $\{(T_{i},Y_{i},X_{i}):i=1,\dots,N\}$
are also independent and identically distributed. For simplicity,
we drop the subscript $i$ when no confusion arises. 

We focus on the setting satisfying treatment ignorability in observational
studies \citep{rosenbaum1983central}.
\begin{assumption}[No unmeasured confounding]\label{assump:TAignorability}
$\{Y(1),Y(0)\}\ind T\mid X.$ \end{assumption}
Assumption~\ref{assump:TAignorability} rules out latent confounding
between the treatment and outcome. In observational studies, its plausibility
relies on whether or not the observed covariates $X$ include all
the confounders that affect the treatment as well as the outcome.

Most of the existing works \citep{nie2017quasi,wager2018estimation,kennedy2020optimal,semenova2020debiased}
focus on estimating the CATE given the full set of $X$, i.e., $\gamma(x):=\E\{Y(1)-Y(0)\mid X=x\}$,
$x\in\mathcal{X}$, which we refer to as the FCATE.
However, to ensure Assumption~\ref{assump:TAignorability} holds, $X$ is often
multidimensional, leading to a multidimensional CATE function $\gamma(x)$
that is challenging to estimate. Indeed, it is common that some covariates
in $X$ are simply confounders but not treatment effect modifiers
of interest. Therefore, a more sensible way is to allow the conditioning
variables to be a subset of confounders \citep{abrevaya2015estimating,zimmert2019nonparametric,fan2020estimation}.
Instead of $\gamma(x)$, we focus on estimating the PCATE
\[
\tau(v)=\E\left\{ Y(1)-Y(0)\mid V=v\right\} ,\quad v\in\mathcal{V}\subseteq\mathcal{X},
\]
where $V$ is a subset of $X$.
It is worth noting that $V=X$ is also allowed, and therefore $\gamma(x)$ can be estimated by our framework.
For simplicity, we assume $V$ is a continuous random vector for the rest of
the paper. 
When $V$ contains discrete random variables, one can divide the sample into different strata, of which the units have the same level of discrete covariates. Then $\tau(v)$ can be estimated by estimating the PCATE at every strata.

In addition to Assumption \ref{assump:TAignorability}, we require sufficient overlap between the treatment groups. Let $\pi(x)=\pr(T=1\mid X=x)$ be the propensity score. Throughout
this paper, we also assume that the propensity score is strictly bounded above
zero and below one to ensure overlap.
\begin{assumption} \label{assum: prop} The propensity score $\pi(\cdot)$
is uniformly bounded away from zero and one. That is, there exist a constant
$\newu\ltxlabel{prop}>0$, such that 
$1/  \oldu{prop} \le \pi(x) \le (1-1/ \oldu{prop})$
for all $x\in\mathcal{X}$.
\end{assumption}

Under Assumptions~\ref{assump:TAignorability} and \ref{assum: prop},
$\tau(v)$ is identifiable based on the following formula
\[
\tau(v)=\E\left\{ Y(1)-Y(0)\mid V=v\right\} =\E\left\{ \left.\frac{TY}{\pi(X)}-\frac{(1-T)Y}{1-\pi(X)}~\right\vert ~V=v\right\} .
\]
First suppose $\pi(X_i), i=1,\dots,N,$ are known.  Common procedures  construct adjusted responses $Z_{i}=T_{i}Y_{i}/\pi(X_{i})-(1-T_{i})Y_{i}/\{1-\pi(X_{i})\}$  and apply kernel smoother to the data
$\{(V_{i},Z_{i}),i=1,\dots,N\}$. Specifically, let $K(v)$ be a kernel function
and $h>0$ be a bandwidth parameter (with technical conditions specified in Section \ref{sec:assum}). The above strategy leads to the
following estimator for $\tau(v)$: 
\begin{align}
\frac{1/(Nh^{d_{1}})\sum_{i=1}^{N}K\left\{ (V_{i}-v)/h\right\} Z_{i}}{1/(Nh^{d_{1}})\sum_{j=1}^{N}K\left\{ (V_{j}-v)/h\right\} }=\frac{1}{N}\sum_{i=1}^{N}\tK(V_{i},v)Z_{i}\label{eqn:cate}
\end{align}
where
\[
\tK(v_{1},v_{2})=\frac{\frac{1}{h^{d_{1}}}K\{(v_{1}-v_{2})/h\}}{\frac{1}{N}\sum_{j=1}^{N}\frac{1}{h^{d_{1}}}K\{(V_{j}-v_{2})/h\}}.
\]
In observational studies, the propensity scores $\pi(X_i),i=1,\dots,N$,
are often unknown. \citet{abrevaya2015estimating} propose to estimate
these scores using another kernel smoother, and
construct the adjusted responses based on the estimated propensity scores.
There are two drawbacks with this approach.
First, it is well known that
inverting the estimated propensity scores can result in instability,
especially when some of the estimated propensity scores are close
to zero or one. Second, this procedure relies on the propensity score model to
be correctly specified or sufficiently smooth to approximate well. 

To overcome these issues, instead of obtaining the weights by inverting
the estimated propensity scores, we focus on estimating the proper
weights directly. In the next section, we adopt the idea of covariate
balancing weighting, which has been recently studied in the context of average treatment effect (ATE) estimation \citep[e.g.,][]{hainmueller2012entropy, imai2014covariate, zubizarreta2015stable, chan2016globally,  wong2018kernel, zhao2019covariate,kallus2020generalized, wang2020minimal}.

\section{Covariate function balancing weighting for PCATE estimation}
\label{sec: estimator}
\subsection{Motivation}

To motivate the proposed estimator, suppose we are given the covariate
balancing weights $\{\hat{w}_{i}:i=1,\dots,N\}$. We express the adjusted
response as
\begin{align}
Z_{i}=\hat{w}_{i}T_{i}Y_{i}-\hat{w}_{i}(1-T_{i})Y_{i},\quad i=1,\dots,N.\label{eqn:adjusted}
\end{align}
Combining \eqref{eqn:cate} and \eqref{eqn:adjusted}, the 
estimator of $\tau(v)$ is 
\begin{align}
\hat{\tau}(v)=\frac{1}{N}\sum_{i=1}^{N}T_{i}\hat{w}_{i}\tK(V_{i},v)Y_{i}-\frac{1}{N}\sum_{i=1}^{N}(1-T_{i})\hat{w}_{i}\tK(V_{i},v)Y_{i}.\label{eqn:cate_est}
\end{align}
One can see that the estimator \eqref{eqn:cate_est} is a difference
between two terms, which are the estimates of $\mu_{1}(v)=\E\{Y(1)\mid V=v\}$
and $\mu_{0}(v)=\E\{Y(0)\mid V=v\}$, respectively. For simplicity, we focus on
the first term and discuss the estimation of the corresponding weights
$\{w_{i}:T_{i}=1\}$ in the treated group. The same discussion applies
to the second term and the estimation of weights in the control group.

We assume $Y_{i}(1)=m_{1}(X_{i})+\varepsilon_{i}$ such that the $\varepsilon_{i}$'s
are independent random errors with $\E(\varepsilon_{i}\mid X_{i})=0$
and $\E(\epsilon_{i}^{2}\mid X_{i})\leq\sigma_{0}^{2}<\infty$. 
Focusing on the first term of \eqref{eqn:cate_est}, we obtain the
following decomposition 
\begin{align}
\begin{split}\frac{1}{N}\sum_{i=1}^{N}T_{i}\hat{w}_{i}\tK(V_{i},v)Y_{i} & =\frac{1}{N}\sum_{i=1}^{N}T_{i}\hat{w}_{i}\tK(V_{i},v)m_{1}(X_{i})+\frac{1}{N}\sum_{i=1}^{N}T_{i}\hat{w}_{i}\tK(V_{i},v)\varepsilon_{i}\\
 & =\frac{1}{N}\sum_{i=1}^{N}(T_{i}\hat{w}_{i}-1)\tK(V_{i},v)m_{1}(X_{i})+\frac{1}{N}\sum_{i=1}^{N}T_{i}\hat{w}_{i}\tK(V_{i},v)\varepsilon_{i}\\
 & \quad+\left[\frac{1}{N}\sum_{i=1}^{N}\tK(V_{i},v)m_{1}(X_{i})-\mu_{1}(v)\right]+\mu_{1}(v).
\end{split}
\label{eqn:decom}
\end{align}
In the last equality, only the first two terms depend on the weights.
The second term $N^{-1}\sum_{i=1}^{N}T_{i}\hat{w}_{i}\tK(V_{i},v)\varepsilon_{i}$
will be handled by controlling the variability of the weights. The
challenge lies in controlling the first term, which requires the control
of the (empirical) balance of a kernel-weighted function class because
$m_{1}(X_{i}), i=1,\dots, N$, are unknown. 
This requirement makes
achieving covariate balance 
significantly more challenging than
those for estimating the ATE, \emph{i.e.}, when $V$ is deterministic \citep[e.g.,][]{hainmueller2012entropy, imai2014covariate, zubizarreta2015stable, chan2016globally,  wong2018kernel, zhao2019covariate,kallus2020generalized, wang2020minimal},
for multiple reasons: (i) covariate balance is required for all $v$
in a continuum, and (ii) the bandwidth $h$ in kernel smoothing is required to diminish 
with respect to the sample size $N$.

\subsection{Balancing via empirical residual moment operator\label{sec:estimation} }

Suppose $m_{1}\in\mathcal{H}$, where $\mathcal{H}$ is an RKHS with
reproducing kernel $\kernel$ and norm $\|\cdot\|_{\mathcal{H}}$.
Also, let the squared empirical norm be $\|u\|_{N}^{2}=(1/N)\sum_{i=1}^{N}\{u(X_{i})\}^{2}$
for any $u\in\mathcal{H}$. Intuitively, from the first term of $(\ref{eqn:decom})$,
we aim to find weights $w=\{w_{i}:T_{i}=1\}$ to ensure the following
function balancing criteria: 
\begin{align*}
\frac{1}{N}\sum_{i=1}^{N}T_{i}\hat{w}_{i}u(X_{i})\tK(V_{i},v)\approx\frac{1}{N}\sum_{i=1}^{N}u(X_{i})\tK(V_{i},v),
\end{align*}
for all $u\in\mathcal{H},$ where the left and right hand sides are regarded as functions
of $v$. To quantify such an approximation, we define the operator
$\mathcal{M}_{N,h,w}$ mapping an element of $\mathcal{H}$ to a function
on $\mathcal{V}$ by 
\[
\mathcal{M}_{N,h,w}(u,\cdot)=\frac{1}{N}\sum_{i=1}^{N}(T_{i}w_{i}-1)u(X_{i})\tK(V_{i},\cdot),
\]
which we call the empirical residual moment operator with respect to the weights
in $w$. 
The approximation and hence the balancing error can be measured by
\begin{align}
\|\mathcal{M}_{N,h,w}(u,\cdot)\|^{2},
\end{align}
where $\|f\|$ is a generic metric applied to a function $f$ defined
on $\mathcal{V}$.
Typical examples of a metric are $L_{\infty}$-norm ($\|\cdot\|_{\infty}$), $L_{2}$-norm
($\|\cdot\|_{2}$) and empirical norm ($\|\cdot\|_{N}$). If one has
non-uniform preference over $\mathcal{V}$, weighted $L_{2}$-norm
and weighted empirical norm are also applicable. 
In the following, we focus on the balancing error based on $L_{2}$-norm:
\begin{align}
S_{N,h}(w,u) & =\|\mathcal{M}_{N,h,w}(u,\cdot)\|_{2}^{2}.
\label{eqn:S_form}
\end{align}
We will return to the discussion of other norms in Section \ref{sec:theory}.
Ideally, our target is to minimize $\sup_{u\in\Hscr}S_{N,h}(w,u)$
uniformly over a sufficiently complex space $\Hscr$. As soon as one
attempts to do this, one may find that $S_{N,h}(w,tu)=t^{2}S_{N,h}(w,u)$
for any $t\ge0$, which indicates a scaling issue about $u$. Therefore,
we will standardize the magnitude of $u$ and restrict the space to
$\Hscr_{N}=\{u\in\Hscr:\|u\|_{N}^{2}=1\}$ as in \citet{wong2018kernel}.
Also, to overcome  overfitting, 
we add a penalty on $u$ with respect to $\|\cdot\|_{\Hscr}$ and
focus on controlling the balancing error over smoother functions.
Inspired by the discussion for \eqref{eqn:decom}, we also introduce
another penalty term 
\begin{align}
R_{N,h}(w)=\frac{1}{N}\sum_{i=1}^{N}\|T_{i}w_{i}\tilde{K}_{h}(V_{i},\cdot)\|_{2}^{2},\label{eqn:obj_V}
\end{align}
to control the variability of
the weights.

In summary, given any $h>0$, our CFB weights $\hat{w}$ is constructed
as follows: 
\begin{align}
\hat{w}=\argmin_{w}\left[\sup_{u\in\Hscr_{N}}\left\lbrace S_{N,h}(w,u)-\lambda_{1}\|u\|_{\Hscr}^{2}\right\rbrace +\lambda_{2}R_{N,h}(w)\right],\label{eqn:obj_all}
\end{align}
where $\lambda_{1}$ and $\lambda_{2}$ are tuning parameters ($\lambda_{1}>0$ and $\lambda_{2}>0$).
Note that \eqref{eqn:obj_all} does not depend on the weights $\{w_{i}:T_{i}=0\}$
of the control group, and the optimization is only performed with
respect to $\{w_{i}:T_{i}=1\}$.

\begin{remark}
By standard representer theorem, we can show that the solution $\tilde{u} = \hat{u} /\|\hat{u}\|_N$ of the inner optimization satisfies that $\hat{u}$ belongs to $\mathcal{K}_N=\mathrm{span}\{\kernel(X_{i},\cdot):i=1,\dots,N\}$.
See also \supp{\ref{sec:reparameterization}}.
Therefore, by the definition of $\mathcal{M}_{N,h,w}$, the weights are determined by achieving balance of the covariate functions generated by the hybrid of the two reproducing kernel $\kappa$ and the smoothing kernel $K$.
\end{remark}

\begin{remark}
\label{rmk:ate_rkhs}
\cite{wong2018kernel}  adopt a similar optimization form as in \eqref{eqn:obj_all} to obtain  weights.
The key difference between their estimator and ours is the choice of balancing error
tailored to the target quantity.
In \cite{wong2018kernel}, the choice of balancing error is
$\{\sum_{i=1}^N(T_iw_i-1)u(X_i)/N\}^2$, which is designed for estimating the {\it scalar} ATE.
There is no guarantee that the resulting weights will ensure  enough balance for
the estimation of the PCATE, {\it a function of $v$}.
Heuristically, one can regard the balancing error in \cite{wong2018kernel} as
the limit of $S_{N,h}$ as $h\rightarrow\infty$.
For finite $h$, two fundamental difficulties emerge that do not exist in \cite{wong2018kernel}.
First, $\mathcal{M}_{N,h,w}(u,v)$ changes with $v$, and so the choice of $S_{N,h}$
involves a metric for a function of $v$ in \eqref{eqn:S_form}.
This is directly related to the fact that our target is a function (PCATE) instead of a scalar (ATE).
For reasonable metrics, the resulting balancing errors measure imbalances over
all (possibly infinite) values of $v$, which is significantly more difficult than the imbalance control required for ATE.
Second, for each $v$, the involvement of kernel function in $\mathcal{M}_{N,h,w}$
suggests that the effective sample size used in the corresponding balancing is much smaller than $\sum^N_{i=1} T_i$.
There is no theoretical guarantee for the weights of \cite{wong2018kernel} 
to ensure enough balance required for the PCATE,
since the proposed weights are designed to balance a function instead of a scalar. 
We show that the proposed CFB weighting estimator achieves desirable properties both theoretically (Section \ref{sec:theory}) and empirically (Section \ref{sec:sim}). 
\end{remark}

\subsection{Computation}\label{sec:com}

Applying the standard representer theory,
\eqref{eqn:obj_all} can be reformulated as
\begin{align}
\hat{w}=\argmin_{w\ge1}\left[\sigma_{1}\left\lbrace \frac{1}{N}P^{\tp}\mathrm{diag}(T\circ w-J)G_{h}\mathrm{diag}(T\circ w-J)P-N\lambda_{1}D^{-1}\right\rbrace +\lambda_{2}R_{N,h}(w)\right],\label{eqn:obj_final}
\end{align}
where $\circ$ is the element-wise
product of two vectors, $J= (1,1,\dots,1)^\tp$, $\sigma_{1}(A)$ represents the maximum eigenvalue of a symmetric
matrix $A$,  $P \in \mathbb{R}^{N\times r}$ consists of the singular vectors of gram matrix $M:=[\kernel(X_{i},X_{j})]_{i,j=1}^{N}\in\mathbb{R}^{N\times N}$ of rank $r$, $D\in \mathbb{R}^{r\times r}$ is the diagonal matrix such that $M = P D P^\tp$, and 
\[
G_{h}=\left[\begin{array}{ccc}
\int_{\mathcal{V}}\tilde{K}_{h}(V_{1},v)\tilde{K}_{h}(V_{1},v)\de v & \cdots & \int_{\mathcal{V}}\tilde{K}_{h}(V_{1},v)\tilde{K}_{h}(V_{N},v)\de v\\
\vdots & \ddots & \vdots\\
\int_{\mathcal{V}}\tilde{K}_{h}(V_{N},v)\tilde{K}_{h}(V_{1},v)\de v & \cdots & \int_{\mathcal{V}}\tilde{K}_{h}(V_{N},v)\tilde{K}_{h}(V_{N},v)\de v
\end{array}\right]\in\mathbb{R}^{N\times N}.
\]
The detailed derivation can be found in \supp{\ref{sec:reparameterization}}.

As for the computation, Lemma \ref{lemma:convex}, whose proof can be found in \supp{\ref{sec:proof_convex}}, indicates that the underlying optimization is convex.
\begin{lemma} \label{lemma:convex}
The optimization \eqref{eqn:obj_all} is convex.
\end{lemma}
Therefore, generic convex optimization algorithms are applicable.
We note that the corresponding gradient
has a closed-form expression\footnote{when the maximum eigenvalue in the objective function is of multiplicity 1}. Thus, gradient based algorithms can be applied efficiently to solve this problem.

Next we discuss several practical strategies to speed up the optimization.
When optimizing \eqref{eqn:obj_final},
we need to compute
the dominant eigenpair
of an $r\times r$ matrix (for computing the gradient).
Since common choices of the reproducing kernel
$\kernel$ are smooth, the corresponding Gram matrix $M$
can be approximated well by a low-rank matrix. When $N$ is large, to facilitate
computation, one can use an $M$ with $r$ much smaller than $N$,
such that the eigen decomposition of gram matrix $M$ 
approximately holds.
This would significantly reduce the burden of computing the dominant eigenpair of the $r\times r$
matrix.
Although the form of $G_{h}$
may seem complicated, this does not change with $w$.
Therefore, for each $h$, we can precompute $G_{h}$ once
at the beginning of an algorithm for the optimization \eqref{eqn:obj_final}.
However, when the integral $g_h(v_1, v_2) = \int_{\mathcal{V}} \tilde{K}_h(v_1, v) \tilde{K}_h(v_2, v) dv $ does not possess a known expression, one generally has to perform a large number of numerical integrations
for the computation of $G_h$, when $N$ is large.
But, for smooth choices of $K$,
$g_h$ is also a smooth function.
When $N$ is large, we could evaluate $g_h(V_i,V_j)$, $i \in S_1, j \in S_2$ at smaller subsets $S_1$ and $S_2$. Then typical interpolation methods \citep{harder1972interpolation} can be implemented to approximate unevaluated integrals in $G_{h}$
to ease the computation burden.

\section{Augmented estimator }
\label{sec:aug_estimator}
Inspired by the augmented inverse propensity weighting (AIPW) estimators
in the ATE literature, we also propose an augmented estimator
that directly adjusts for the outcome models $m_{1}(\cdot)$ and $m_{0}(\cdot)$.

Recall that the outcome regression functions $m_{1}(\cdot)$ and $m_{0}(\cdot)$ are assumed
to be in an RKHS $\Hscr$, kernel-based estimators $\hat{m}_{1}(\cdot)$
and $\hat{m}_{0}(\cdot)$ can be employed. We can then perform augmentation and obtain the
adjusted response $Z_{i}$ in \eqref{eqn:adjusted} as
\begin{align}
Z_{i}=w_{i}T_{i}\{Y_{i}-\hat{m}_{1}(X_{i})\}+\hat{m}_{1}(X_{i})-\left[w_{i}(1-T_{i})\left\{ Y_{i}-\hat{m}_{0}(X_{i})\right\} +\hat{m}_{0}(X_{i})\right].
\end{align}
Correspondingly, the decomposition in \eqref{eqn:decom} becomes

\begin{align*}
 & \frac{1}{N}\sum_{i=1}^{N}\tilde{K}_{h}(V_{i},v)\hat{m}_{1}(X_{i})+\frac{1}{N}\sum_{i=1}^{N}T_{i}\hat{w}_{i}\tilde{K}_{h}(V_{i},v)\{Y_{i}-\hat{m}_{1}(X_{i})\}\\
= & \frac{1}{N}\sum_{i=1}^{N}(1-T_{i}\hat{w}_{i})\tilde{K}_{h}(V_{i},v)\hat{m}_{1}(X_{i})+\frac{1}{N}\sum_{i=1}^{N}T_{i}\hat{w}_{i}\tilde{K}_{h}(V_{i},v)m_{1}(X_{i})+\frac{1}{N}\sum_{i=1}^{N}T_{i}\hat{w}_{i}\tilde{K}_{h}(V_{i},v)\epsilon_{i}\\
 =& \frac{1}{N}\sum_{i=1}^{N}(T_{i}\hat{w}_{i}-1)\tilde{K}_{h}(V_{i},v)\{m_{1}(X_{i})-\hat{m}_{1}(X_{i})\}+\frac{1}{N}\sum_{i=1}^{N}T_{i}\hat{w}_{i}\tilde{K}_{h}(V_{i},v)\epsilon_{i}\\
 & +\left\{ \frac{1}{N}\sum_{i=1}^{N}\tilde{K}_{h}(V_{i},v)m_{1}(X_{i})-\mu_{1}(v)\right\} +\mu_{1}(v).
\end{align*}

Now, our goal is to control the difference between $N^{-1}\sum_{i=1}^{N}T_{i}\hat{w}_{i}\tilde{K}_{h}(V_{i},v)\{m_{1}(X_{i})-\hat{m}_{1}(X_{i})\}$
and $N^{-1}\sum_{i=1}^{N}\tilde{K}_{h}(V_{i},v)\{m_1(X_i)-\hat{m}_{1}(X_{i})\}$. The weight estimators in Section \ref{sec:estimation} can be adopted similarly to control this difference. It can be shown that the term $S_{N,h}(\hat w, m_1-\hat m_1):=\|\sum_{i=1}^{N}(T_{i}\hat{w}_{i}-1)\tilde{K}_{h}(V_{i},\cdot)\{m_{1}(X_{i})-\hat{m}_{1}(X_{i})/N\}\|_2^2$ can achieve a faster rate of convergence than $S_{N,h}(\hat w, m_1)$ does with the same estimated weights $\hat w$ as long as $\hat m_1$ is a consistent estimator.
However, this property does not improve 
the final convergence rate of the PCATE estimation.
This is because the term
$\| {N}^{-1}\sum_{i=1}^{N}\tilde{K}_{h}(V_{i},\cdot)m_{1}(X_{i})-\mu_{1}(\cdot)\|_2^2$ dominates other terms, and thus
the final rate can never be faster than the optimal non-parametric rate. See Remark \ref{rmk: aug} for more details.
Our theoretical results  reveal that
the benefit of using the augmentations lies in the relaxed order requirement of the tuning parameters
to achieve the optimal convergence rate. Therefore,
the performance of the augmented estimator is expected to be more robust to the tuning parameter selection.

Unlike other AIPW-type estimators \citep{lee2017doubly,fan2020estimation,zimmert2019nonparametric,semenova2020debiased}
which often rely on data splitting for estimating the propensity score
and outcome mean functions to relax technical conditions, 
our estimator does not
require data splitting to facilitate the convergence with augmentation. See also Remark \ref{rmk:splitting}. We  defer the theoretical comparison between our estimator and the existing AIPW-type estimators in Remark \ref{rmk:aipw2} in Section \ref{sec:theory}.

Last, we note that there are existing work using weights to balance the residuals \citep[e.g.][]{athey2016approximate,wong2018kernel}, which is similar to what we consider here for the proposed augmented estimator.  These estimators are designed for ATE estimation and the balancing weights cannot be directly adopted here with theoretical guarantee.

\section{Asymptotic properties}
\label{sec:theory}

In this section, we conduct an asymptotic analysis for the proposed estimator. For simplicity,
we assume $\mathcal{X}=[0,1]^d$. To facilitate our theoretical discussion in terms of smoothness,
we assume the RKHS $\mathcal{H}$ is contained in a Sobolev space (see Assumption \ref{assum: ratio}).
Our results can be extended to other choices of $\Hscr$ if the corresponding entropy result and boundedness condition for the unit ball $\{u\in \Hscr: \|u\|_\Hscr\leq 1\}$ are provided.
Recall that we focus on $\E\{ Y(1)\mid V = v\}$. Similar analysis can be applied to  $\E\{ Y(0)\mid V =v \}$ and finally the PCATE.

\subsection{Regularity conditions}
\label{sec:assum}

Let $\ell$ be a positive integer.
For any function $u$ defined on $\mathcal{X}$,  the the Sobolev norm is
$\|u\|_{\mathcal{W}^\ell} = \sqrt{\sum_{|\beta|\le \ell} \|D^\beta u\|_{2}^2}$,
where $D^\beta u(x_1,\dots,x_d) = \frac{\partial^{|\beta|}u}{\partial x_1^{\beta_1} \dots \partial x_d^{\beta_d}}$ for a multi-index $\beta=(\beta_1,\dots,\beta_d)$.
The Sobolev space $\mathcal{W}^\ell$ consists of functions with finite Sobolev norm.
For $\epsilon>0$,  we denote by $\mathcal{N}(\epsilon, \Fscr, \|\cdot\|)$ the $\epsilon$-covering number of a set $\Fscr$ with respect to some norm $\|\cdot\|$. Next, we list the assumptions that are useful for our
asymptotic results.

\begin{assumption}
	\label{assum: ratio}
    The unit ball of $\mathcal{H}$ is a subset of a ball in the Sobolev space $\mathcal{W}^\ell$,
    with the ratio $\alpha := d/\ell$ less than 2.
\end{assumption}

\begin{assumption}
	\label{assum: m}
	The regression function $m(x) \in \Hscr$.
\end{assumption}

\begin{assumption}
	\textbf{(a)} K is symmetric, $\int K(s) ds = 1$, and there exists a constant $\newu\ltxlabel{supkernel}$ such that $K(s) \leq \oldu{supkernel}$ for all $s$. Moreover, 
	$\int s^2 K(s) ds < \infty$ and $\int K^2(s)ds<\infty$.
	\textbf{(b)} Take $\Kscr = \{ K\{(v-\cdot)/h\} : h>0, v\in[0,1]^{d_1}\}$. There exist constants $\Ak>0$ and $\nuk >0$ such that $\mathcal{N}(\varepsilon, \Kscr, \|\cdot\|_\infty) \leq \Ak \varepsilon^{-\nuk}$.
	\label{assum: kernel_2}
\end{assumption}

\begin{assumption}
	\label{assum: density}
	The density function $g(\cdot)$ of the random variable $V\in[0,1]^{d_1}$ is continuous, differentiable, and  bounded away from zero, \emph{i.e.},
	there exist constants $\newu\ltxlabel{density_lower}>0$ and $\newu\ltxlabel{density_upper}>0$ such that
	$\oldu{density_lower}\leq g(v)  \leq \oldu{density_upper}$.
\end{assumption}

\begin{assumption}
	\label{assum: ban}
	$h \rightarrow 0$ and  $N^{\frac{2}{2+\alpha}}h^{d_1} \rightarrow \infty$, as $N\rightarrow \infty$. 
\end{assumption}

\begin{assumption}
	\label{assum: joint}
The joint density of $\{m(X), V\}$ and the conditional expectation $\E\{m(X) \mid V = v\}$ are continuous.
\end{assumption}

\begin{assumption}
	\label{assum:error}
	The errors $\{\varepsilon_i, i= 1,\dots, N\}$ are uncorrelated, with $\E(\varepsilon_i) = 0$ and $\Var(\varepsilon_i) \leq \sigma_0^2$ for all $i = 1, \dots, N$.  Furthermore,  $\{\varepsilon_i, i= 1,\dots, N\}$ are independent of $\{T_i, i = 1,\dots, N\}$ and $\{X_i, i =1, \dots, N\}$.
\end{assumption}

Assumption \ref{assum: ratio} is a common condition in the literature of smoothing spline regression. Assumptions \ref{assum: kernel_2}--\ref{assum: joint} comprise standard conditions for kernel smoother \citep[e.g.,][]{mack1982weak,einmahl2005uniform,wasserman2006all} except that we require $N^{\frac{\alpha}{2+\alpha}}h^{d_1} \rightarrow \infty$ instead of $Nh^{d_1} \rightarrow \infty$ to
ensure the difference between $\|u\|_N$ and $\|u\|_2$ is asymptotically negligible.
Assumption \ref{assum: kernel_2}(b) is satisfied whenever
$K(\cdot) = \psi\{p(\cdot)\}$ with $p(\cdot)$ being a polynomial in $d_1$ variables and $\psi$ being a real-valued function of bounded variation \citep*{van2000asymptotic}. 

\subsection{$L_2$-norm balancing}\label{sec: L2theory}

Given two sequences of positive real numbers $(A_1,A_2,\dots)$ and $(B_1,B_2,\dots)$,
$A_N = \bigO(B_N)$ represents that there exists a positive constant $M$ such that 
$A_N \leq M B_N$ as $N \rightarrow \infty$; $A_N = \smallO(B_N)$ represents that $A_N /B_N \rightarrow 0$ as $N\rightarrow \infty$, and $A_N \asymp B_N$ represents $A_N = \bigO(B_N)$ and $B_N = \bigO(A_N)$. 

\begin{thm}
	\label{thm: prep}
	Suppose Assumptions \ref{assump:TAignorability}--\ref{assum: ban} hold.
	If $\lambda_1^{-1} = \smallO(Nh^{d_1})$, we have $S_{N,h}(\hat w, m) = \bigOp (\lambda_1\|m\|_N^2 + \lambda_1\|m\|_\Hscr^2 + \lambda_2 h^{-d_1}\|m\|_N^2) $. If we further assume $\lambda_2^{-1} = \bigO(\lambda_1^{-1} h^{-d_1}) $,  then $R_{N,h}(\hat w) = \bigOp(h^{-d_1})$.
\end{thm}

Theorem \ref{thm: prep} specifies the control of the balancing error and the weight variability. They can be used to derive the convergence rate of the proposed estimator in the following theorem.

\begin{thm}
	\label{thm: decomp}
	Suppose Assumptions \ref{assump:TAignorability}-\ref{assum:error} hold. If  $\lambda_1^{-1} = \smallO(Nh^{d_1})$, $\lambda_2^{-1} = \bigO(\lambda_1^{-1} h^{-d_1})$, and $h^2 = \smallO((N^{-1}h^{-d_1})^{1/2})$,
	\begin{multline*}
		\left\|\frac{1}{N}\sum_{i=1}^{N} T_i\hat w_iY_i K_h\left(  V_i, \cdot\right) -
		\E\left\lbrace Y(1) | V= \cdot\right\rbrace
		\right\|_2 \\ = \bigOp\{N^{-1/2}h^{-d_1/2}+\lambda_1^{1/2}\|m_1\|_\Hscr + \lambda_2^{1/2} h^{-d_1/2}\|m_1\|_2\}.
	\end{multline*}
\end{thm}
The proof can be found in \supp{\ref{sec:proof_prep} and \ref{sec:proof_decomp}}. 
Since we require $\lambda_1^{-1} = \smallO(Nh^{d_1})$, the best convergence rate that we can achieve in Theorem \ref{thm: decomp} is arbitrarily close to the optimal rate $N^{-1/2}h^{-d_1/2}$.
It is unclear if this arbitrarily small gap is an artifact of our proof structure.
However, in Theorem \ref{thm: aug} below, we show that this gap can be closed by using the proposed augmented estimator.

\begin{remark}
	\label{rmk:ipw}
	\cite{abrevaya2015estimating} adopt an inverse probability weighting (IPW) method to estimate the PCATE, where the propensity scores are approximated parametrically or by kernel smoothing. 
	They provide point-wise convergence result for their estimators, as opposed to $L_2$ convergence in our theorem.
	For their nonparametric propensity score estimator, their result is derived based on a strong smoothness assumption of the propensity score. 
	More specifically, it requires high-order kernels (the order should not be less than $d$) in estimating both the propensity score and the later PCATE  in order to achieve the optimal convergence rate.
	Compared to their results, our proposed estimator does not involve such a strong smoothness assumption nor a parametric specification of the propensity score.
\end{remark}

\subsection{$L_\infty$-norm balancing}
In Section \ref{sec:estimation},
we mention several choices of the metric in the balancing error \eqref{eqn:S_form}.
In this subsection, we provide a theoretical investigation of an important
 case with $L_\infty$-norm.
We note that efficient computation of the corresponding weights
is challenging, and thus is not pursued in the current paper. 
Nonetheless, it is theoretically interesting to derive the convergence result for the proposed estimator with $L_\infty$-norm.  
More specifically,
the estimator of interest in this subsection is defined by replacing the $L_2$-norm in $S_{N,h}(w, u)$ and $R_{N,h}(w)$ with the $L_\infty$-norm.
Instead of $L_2$ convergence rate (Theorem \ref{thm: decomp}),
we can obtain the uniform convergence rate of this estimator in the following theorem.
\begin{thm}
\label{thm:sup}
  Suppose Assumptions \ref{assump:TAignorability}--\ref{assum:error} hold, 
Let $\tilde{w}$ be the solution to \eqref{eqn:obj_all}
but with $S_{N,h}(w,u) = \|\mathcal{M}_{N,h,w}(u, \cdot)\|_\infty$ and $R_{N,h}(w) = \|\frac{1}{N} \sum_{i=1}^N T_iw_i\tilde{K}_h(V_i,\cdot)\|_\infty$.
If $\lambda_1^{-1} \asymp Nh^{d_1}\log(1/h)$, $\lambda_2 \asymp N^{-1}$, $\log(1/h) /(\log\log N) \rightarrow \infty$ as $N \rightarrow \infty$, and
	$h^2 = \smallO\{(N^{-1}h^{-d_1}\log(1/h))^{1/2}\}$,
$$\left\|\frac{1}{N}\sum_{i=1}^{N} T_i\hat w_iY_i K_h\left(  V_i,\cdot\right) -\E\left\lbrace Y(1) | V= \cdot\right\rbrace  \right\|_\infty  = \bigOp\{N^{-1/2}h^{-d_1/2}\log^{1/2}(1/h)\}.$$
\end{thm}
We provide the proof outline in \supp{\ref{sec:proof_sup}}.

Different from Theorem \ref{thm: decomp}, the  uniform convergence rate is optimal. Roughly speaking, this is because, compared to the optimal $L_2$ convergence rate, the optimal uniform convergence rate has an extra logarithmic order, which dominates the arbitrarily small gap mentioned in Section \ref{sec: L2theory}.

\subsection{Augmented estimator}

We also derive the asymptotic property of the augmented estimator.
\begin{thm}
	\label{thm: aug}
	Suppose Assumptions \ref{assump:TAignorability}--\ref{assum:error} hold. Take $e = m_1 - \hat m_1 \in \Hscr$ such that $\|e\|_\Hscr = \smallOp(1)$ and $\|e\|_2 = \smallOp(1)$.  Suppose  $\lambda_1^{-1} = \smallO(Nh^{d_1})$, 
	$\lambda_2^{-1} = \bigO(\lambda_1^{-1}h^{-d_1})$, and $h^2 = \smallO((N^{-1}h^{-d_1})^{1/2})$, we have
	\begin{align*}
		 & \left\|\frac{1}{N} \sum_{i=1}^N  \tilde K_h (V_i, \cdot) \hat m_1(X_i) + \frac{1}{N} \sum_{i=1}^N T_i \hat w_i \tilde K_h (V_i, \cdot) \{ Y_i - \hat m_1(X_i)\} -\E\left\lbrace Y(1) | V= \cdot\right\rbrace  \right\|_2 \\
		 & =  \bigOp(N^{-1/2}h^{-d_1/2} +  \lambda_1^{1/2} \|e\|_{\Hscr} + \lambda_2^{1/2} h^{-d_1/2} \|e\|_2)
	\end{align*}
\end{thm}

\begin{remark}
	\label{rmk: aug}
	In Theorem \ref{thm: decomp}, to obtain the best convergence rate that is arbitrarily close to $N^{-1/2}h^{-d_1/2}$, we require $\lambda_1$ and $\lambda_2$ to be arbitrarily close to $N^{-1}h^{-d_1}$ and 
	$N^{-1}$ respectively.
	While in Theorem \ref{thm: aug}, as long as $\lambda_1 = \bigO(N^{-1}h^{-d_1}\log(1/h)\|e\|^{-2}_\Hscr)$ and $\lambda_2 = \bigO(N^{-1}\log(1/h)\|e\|^{-2}_N)$,  the optimal convergence rate $N^{-1/2}h^{-d_1/2}$ is achievable.  Therefore,  with the help of augmentation, we can relax the order requirement of the tuning parameters for achieving the optimal rate. As a result,  
	it is ``easier" to tune $\lambda_1$ and $\lambda_2$ with augmentation.
\end{remark}

\begin{remark}
	\label{rmk:aipw1}
	Several existing works focus on estimating the FCATE $\gamma(\cdot)$ given the full set of covariates \citep{kennedy2020optimal, nie2017quasi}.
	While one could partially marginalize their estimate $\hat{\gamma}(\cdot)$ of $\gamma(\cdot)$ to obtain an estimate $\check{\tau}(\cdot)$ of $\tau(\cdot)$,
	it is not entirely clear whether the convergence rate of $\check{\tau}(\cdot)$ is optimal, even when $\hat\gamma(\cdot)$ is rate-optimal non-parametrically.
	The main reason is that the estimation error $\hat{\gamma}(x)-\gamma(x)$ are  dependent across different values of $x$.
	Note that $\gamma(\cdot)$ is a $d$-dimensional function and the optimal rate is  slower than the optimal rate that we achieve for $\tau(\cdot)$, a $d_1$-dimensional function, when $d_1 < d$.
	So the partially marginalizing step needs to be shown to speed up the convergence significantly, in order to be comparable to our rate result.
\end{remark}

\begin{remark}
	\label{rmk:aipw2}
	To directly estimate the PCATE $\tau(\cdot)$, a common approach is
	to apply smoothing methods to the adjusted responses with respect to $V$ instead of $X$.
	Including ours, most papers follow this approach. The essential difficulty discussed in Remark \ref{rmk:aipw1} remains and hence the analyses are more challenging than those for the FCATE $\gamma(\cdot)$,
	if the optimal rate is sought.
	In the existing work \citep{lee2017doubly, semenova2017estimation,zimmert2019nonparametric, fan2020estimation} that adopts augmentation,
	estimations of both propensity score and outcome mean functions, referred to as nuisance parameters in below, are required.
	\cite{lee2017doubly} adopt parametric modeling for both nuisance parameters and achieve
	double robustness; \emph{i.e.}, only one nuisance parameter is required to be consistent to achieve the optimal rate for $\tau(\cdot)$.
	However, parametric modeling is a strong assumption and may be restrictive.
	\citet{semenova2017estimation,zimmert2019nonparametric, fan2020estimation}
	adopt nonparametric nusiance modeling. 
	Importantly, to achieve optimal rate of $\tau(\cdot)$,
	these works require consistency of \textit{both} nuisance parameter estimations.
	In other words, the correct specification of both nuisance parameter models are required.
	\cite{fan2020estimation}
	require  both nuisance parameters to be estimated consistently with respect to $L_\infty$ norm.
	While \citet{semenova2017estimation} and \citet{zimmert2019nonparametric} 
	implicitly require the product convergence rates from the two estimators to be faster than $N^{-1/2}$ 
	to achieve the optimal rate of the PCATE estimation.
	In other words, if one nuisance estimator is not consistent, the other nuisance estimator has to converge faster than $N^{-1/2}$.
	Unlike these existing estimators, 
	our estimators does not rely on restrictive parametric modeling
	nor consistency of both nuisance parameter estimation.
\end{remark}

\begin{remark}
\label{rmk:splitting}
    Moreover, most existing work (discussed in Remark \ref{rmk:aipw2}) require data-splitting or cross-fitting
    to remove the dependence between nuisance parameter estimations and the smoothing step for estimating $\tau(\cdot)$, which is crucial in their theoretical analyses.
    \cite{zheng2011cross} 
     first propose cross-fitting in the context of Target Maximum Likelihood Estimator and \cite{chernozhukov2017double} 
      subsequently apply to estimating equations.  
This technique can be used to relax the Donsker conditions required for the class of nuisance functions.  
   \cite{kennedy2020optimal} applies cross-fitting to FCATE estimation for similar purposes.
    While data-splitting and cross-fitting are beneficial in theoretical development, they are not generally a
    favorable modification, due to criticism of increased computation and fewer data for the estimation of different components (nuisance parameter estimation and smoothing).
    However, our estimators do not require data-splitting in both theory and practice.
	Our asymptotic analyses are non-standard and significantly different than these existing work since, without data-splitting, the estimated weights are intimately related with each others and an additional layer of smoothing further complicates the dependence structure.
\end{remark}

\section{Simulation\label{sec:sim} }

We evaluate the finite-sample properties of various estimators with
sample size $N=100$. 
The covariate $X\in\mathbb{R}^{4}$ is generated
by 
$X_{1}=Z_{1}$, $X_{2}=Z^2_{1}+Z_{2}$, $X_{3}=\exp(Z_{3}/2) + Z_2$ and
$X_{4}= \sin(2Z_1) + Z_4$ 
with $Z_{j}\sim$ Uniform$[-2,2]$ for
$j=1,\ldots,4$. The conditioning variable of interest is set to be
$V=X_{1}$. The treatment is generated by $T\mid X\sim\text{Bernoulli}\{\pi(X)\}$,
and the outcome is generated by $Y\mid(T=t,X)\sim\text{ N}\{m_{t}(X),1\}$.
To assess the estimators, we consider two different
choices for each of $\pi(X)$ and $m_{t}(X)$, summarized in Table~\ref{tab:simulation-setup}.  In Settings 1 and 2, the outcome mean functions $m_t$ are relatively easy to estimate, as they are linear with respect to covariates $X$. While in Settings 3 and 4, the outcome mean functions are nonlinear and more complex. Propensity score function $\pi(X)$ is set to be linear with respect to  $X$ in Settings 1 and 3, and nonlinear in Settings 2 and 4.
The corresponding PCATEs are nonlinear and shown in  Figure \ref{fig:PATEs}. 

\begin{table}[t]
\caption{Models for simulation with two specifications for each of $\logit\{\pi(X)\}$
and $m_{t}(X)$ ($t=0,1$)}
  \centering
  \resizebox{\textwidth}{!}{  
\begin{tabular}{cccc}
\hline 
Setting & $\pi(X)$ & $m_{t}(X)$ $(t=0,1)$ & $\tau(v)$\tabularnewline
\hline 
1 &  $1/(1+ \exp{X_1 + X_3})$ &  $10+ X_1 + (2t-1)(X_{2}+X_{4})$ & $2v^2 + 2\sin(2v) $\tabularnewline
2 &  $1/(1+ \exp{Z_1 +Z_2+ Z_3})$ &   $10+ X_1 + (2t-1)(X_{2}+X_{4})$ & $2v^2 + 2\sin(2v) $\tabularnewline
3 &  $1/(1+ \exp{X_1 + X_3})$ &  $10 + (2t-1)(Z_1^2 + 2Z_1\sin(2Z_1)) + Z_2^2 + \sin(2Z_3)Z_4^2$ & $ 2v^2 + 4v\sin(2v)$\tabularnewline
4 &  $1/(1+ \exp{Z_1 +Z_2+ Z_3})$ & $10 + (2t-1)(Z_1^2 + 2Z_1\sin(2Z_1)) + Z_2^2 + \sin(2Z_3)Z_4^2$ & $ 2v^2 + 4v\sin(2v)$\tabularnewline
\hline 
\end{tabular}
}
\label{tab:simulation-setup}
\end{table}

\begin{figure}
    \centering
    \includegraphics[width=0.8\textwidth]{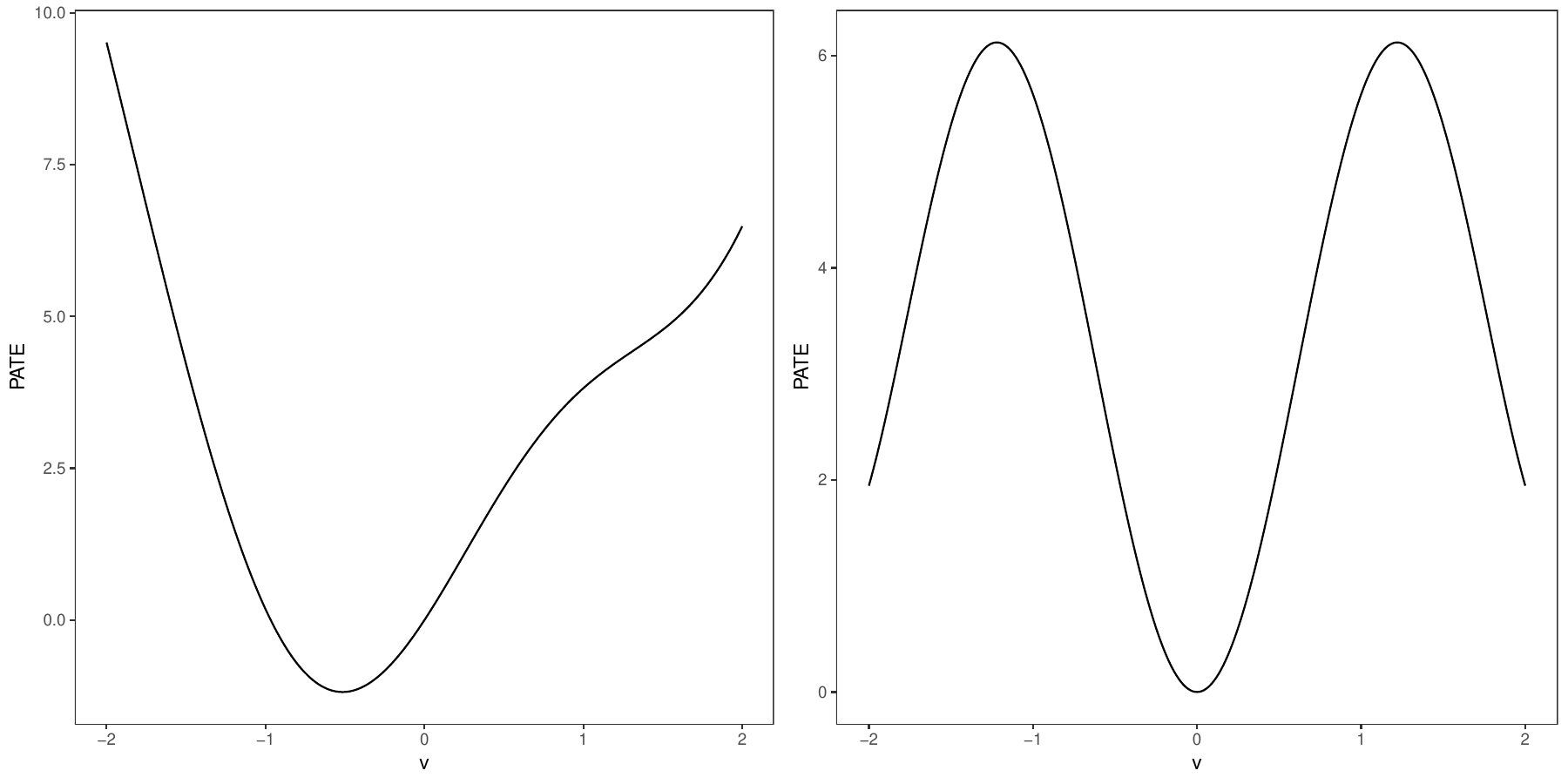}
    \caption{The target PCATEs in the simulation study: the left panel plots the PCATE in Settings 1 and 2; the right panel plots the PCATE in Settings 3 and 4.}
    \label{fig:PATEs}
\end{figure}

In our study, we compare the following estimators for $\tau(\cdot)$:
\begin{enumerate}
\item \cate: the proposed estimator using the tensor product of second-order
Sobolev kernel as the reproducing kernel $\kappa$. 
\item \ate: the weighted estimator described in Remark \ref{rmk:ate_rkhs}, whose weights are estimated based on the covariate balancing criterion in \cite{wong2018kernel}.
\item \ipw: the inverse propensity weighting estimator from \citet{abrevaya2015estimating}
with a logistic regression model for the propensity score. In Settings
1 and 3, the propensity score model is correctly specified. 
\item Augmented estimators by augmenting the estimators in a)--c)
by the outcome models. We consider two outcome models: linear regression
(LM) and kernel ridge regression (KRR). 
\item \Noweights: the estimator that uses outcome regressions. 
It directly smooths $\{(X_i, \hat m_1(X_i) - \hat m_0(X_i)): i=1,\dots,N\}$ to estimate the PCATE, where $\hat m_1(X_i)$ and  $\hat m_0(X_i)$ are estimated with outcome models considered in d).
\end{enumerate}
For all estimators, a kernel smoother with Gaussian kernel is applied to the adjusted
responses. For \ipw, the bandwidth is set
as $\tilde{h}=\hat{h}\times N^{1/5}\times N^{-2/7}$, where $\hat{h}$
is a commonly used optimal bandwidth in the literature such as the
direct plug-in method \citep*{ruppert1995effective,wand1994kernel,calonico2019nprobust}.
Throughout our analysis, $\hat{h}$ is computed via the R package ``nprobust".
The same bandwidth formula $\tilde{h}$ is also considered by \citet{lee2017doubly}
and \citet{fan2020estimation} to estimate the CATE. For the proposed
estimator, a bandwidth should be given prior to estimate
the weights. We first compute the adjusted response by using weights
from \citet{wong2018kernel}, and then obtain the bandwidth $\tilde{h}$ as the
input to our proposed estimator.

\begin{table}[ht]
\centering
\resizebox{\textwidth}{!}{%
\begin{tabular}{cc|cc|cc|cc|cc}
\hline 
 &  & \multicolumn{2}{c}{Setting 1} & \multicolumn{2}{c}{Setting 2 } & \multicolumn{2}{c}{Setting 3 } & \multicolumn{2}{c}{Setting 4 }\tabularnewline
Augmentation  & Method  & AISE  & MeISE  & AISE  & MeISE  & AISE  & MeISE  & AISE  & MeISE \tabularnewline
\hline 
  No & \ipw & 80.212 (16.19) & 25.706 & 40.898 (6.75) & 18.838 & 105.509 (31.66) & 31.146 & 49.04 (9.34) & 20.989 \\ 
    & \ate & 16.136 (0.77) & 9.633 & 9.653 (0.46) & 6.139 & 18.458 (0.93) & 10.264 & 11.367 (0.59) & 7.257 \\ 
    & \cate & 4.223 (0.22) & 2.725 & 2.232 (0.06) & 1.997 & 4.769 (0.26) & 3.229 & 3.214 (0.08) & 3.006 \\ 
    \hline
  LM & \ipw & 1.167 (0.04) & 0.958 & 1.066 (0.03) & 0.893 & 5.431 (1.37) & 2.405 & 3.74 (0.78) & 2.001 \\ 
    & \ate & 1.156 (0.03) & 1.011 & 1.112 (0.03) & 1.003 & 3.471 (0.19) & 2.237 & 2.276 (0.06) & 1.924 \\ 
    & \cate & 1.095 (0.03) & 0.947 & 0.977 (0.02) & 0.868 & 2.966 (0.17) & 2.014 & 1.856 (0.05) & 1.596 \\ 
    & \Noweights & 0.843 (0.03) & 0.716 & 0.767 (0.02) & 0.67 & 5.431 (0.18) & 4.368 & 4.254 (0.05) & 4.107 \\ 
    \hline
  KRR & \ipw & 1.25 (0.04) & 1.048 & 1.039 (0.03) & 0.905 & 3.203 (0.19) & 2.096 & 2.313 (0.14) & 1.645 \\ 
    & \ate & 1.289 (0.04) & 1.092 & 1.152 (0.03) & 0.993 & 3.07 (0.13) & 2.213 & 2.125 (0.06) & 1.827 \\ 
    & \cate & 1.203 (0.04) & 1.023 & 1.012 (0.03) & 0.856 & 2.658 (0.12) & 1.911 & 1.843 (0.05) & 1.53 \\ 
    & \Noweights & 1.137 (0.03) & 0.953 & 0.905 (0.02) & 0.797 & 3.778 (0.12) & 3.213 & 2.796 (0.06) & 2.634 \\ 
\hline 
\end{tabular}} \caption{Simulation results for the four settings, where the average integrated
squared errors (AISE) with standard errors (SE) in parentheses and
median integrated squared error (MeISE) are provided.}
\label{table:simres} 
\end{table}

Table \ref{table:simres} shows the average integrated squared error
(AISE) and median integrated squared error (MeISE) of above estimators
over 500 simulated datasets. Without
augmentation, \cate{} has
 significantly smaller AISE and MeISE than other methods among all four settings.
All methods are improved
by augmentations. In Settings 1 and 2, \Noweights{} has the best performance. In these two settings, the outcome models are linear and thus can be estimated well by both LM and KRR.
However,  the differences between \Noweights{} and \cate{} are relatively small. As for Settings 3 and 4 where outcome mean functions are more complex, \cate{} achieves the best performance and shows a significant improvement over \Noweights{}, especially when outcome models are misspecified (See Settings 3 and 4 with LM augmentation). 
As \ate{} is only designed for marginal covariate
balancing, its performance is worse than \cate{} across all scenarios.

\section{Application}
\label{sec:real}

We apply the estimators in Section \ref{sec:sim} to estimate the
effect of maternal smoking on birth weight as a function of mother's
age, by re-analyzing a dataset of 
mothers
in Pennsylvania in the USA (\url{http://www.stata-press.com/data/r13/cattaneo2.dta}).
Following \citet{lee2017doubly}, we focus on white
and non-Hispanic mothers, resulting in the sample size $N=3754$.
The outcome $Y$ is the infant birth weight measured in grams and
the treatment indicator $T$ is whether the mother is a smoker.
For the treatment ignorability, we include the following
covariates: mother’s age, an indicator variable for alcohol consumption
during pregnancy, an indicator for the first baby, mother’s educational
attainment, an indicator for the first prenatal visit in the first
trimester, the number of prenatal care visits, and an indicator for
whether there was a previous birth where the newborn died.
Due to the boundary effect of the kernel smoother,
we focus on $\tau(v)$ for $v\in [18, 36]$,
which ranges from $0.05$ quantile to $0.95$ quantile of mothers' ages in
the sample.

We compute various estimators of the PCATE in Section \ref{sec:sim}. For all the following \ipw{}  related estimators, logistic regression is adopted to estimate propensity scores. 
Following
\citet{abrevaya2015estimating}, we include IPW: the \ipw{} estimator  with no augmentation.   Following
\citet{lee2017doubly}, we include IPW(LM): 
the \ipw{} estimator with LM augmentation. 
We include Proposed: the proposed estimators
 with KRR augmentation here as it performs the best in the simulation study and aligns with our assumption for the outcome mean functions.  
For completeness, we
also include IPW(KRR): 
the \ipw{} estimator with KRR augmentation; \Noweights{}(KRR): the \Noweights{} estimator where the outcome mean functions are estimated by KRR; \Noweights{}(LM): the \Noweights{} estimator where the outcome mean functions are estimated by LM.
For both the KRR augmentation and the weights estimations in Proposed,
we consider a tensor product RKHS, with the
second order Sobolev space kernel for continuous covariates and the
identity kernel for binary covariates.

\begin{figure}[ht]
    \centering
    \includegraphics[width=\textwidth]{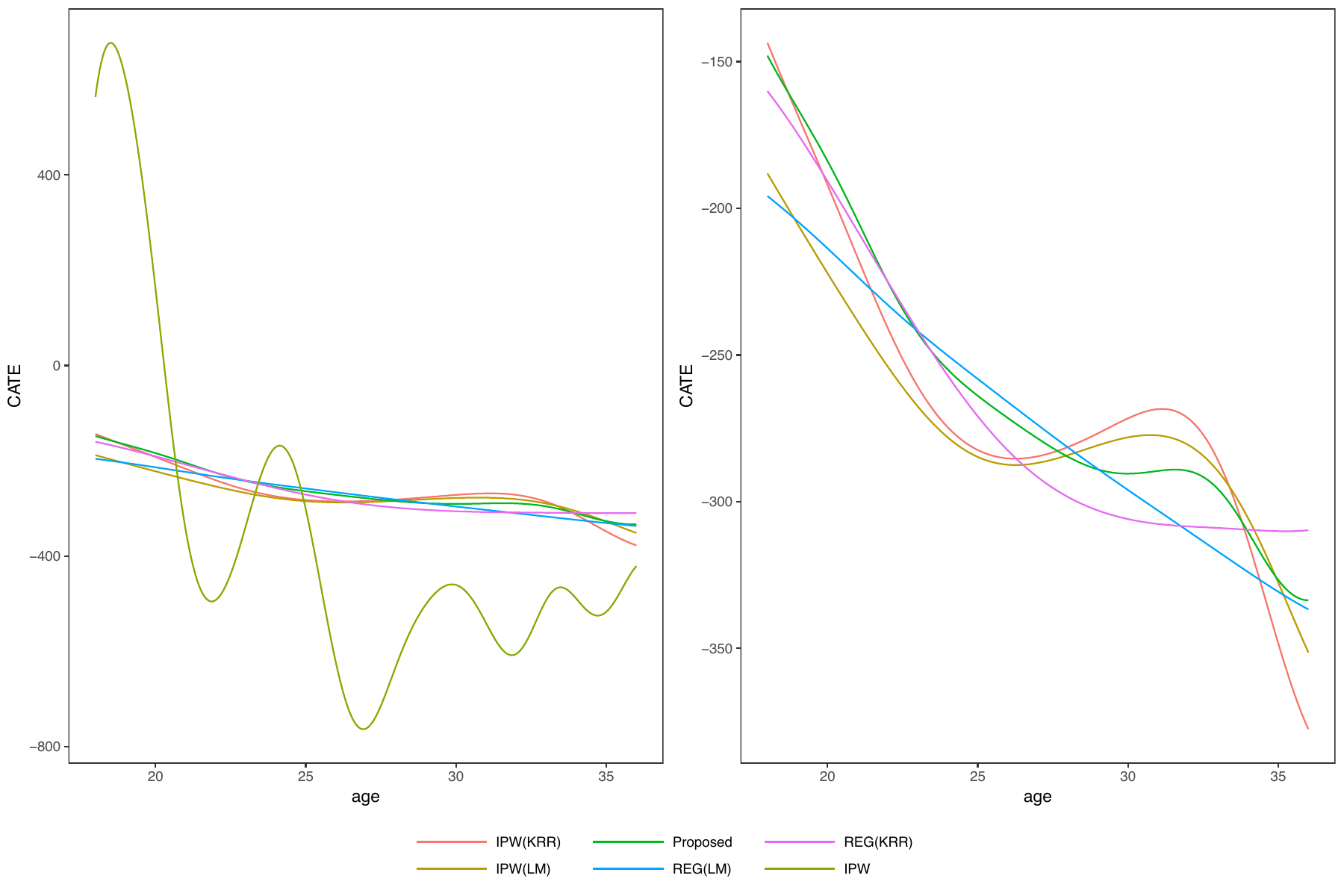}
    \caption{The estimated PCATEs of  maternal smoking on birth weight as a function of mother's
age: the left panel includes all estimators and the right panel excludes
the IPW estimator. }
\label{fig:cate_all}
\end{figure}

Figure \ref{fig:cate_all} shows the estimated PCATEs from different
methods. From the left panel in Figure \ref{fig:cate_all}, IPW has
large variations compared to other estimators. The significantly positive
estimates before age 20 conflict with the results from various
established research works indicating that smoking
has adverse effect on
birth weights \citep{kramer1987intrauterine,almond2005costs,abrevaya2006estimating,abrevaya2008effects}.
From the right panel in Figure \ref{fig:cate_all}, the remaining
four estimators show a similar pattern that the effect becomes more
severe as mother's age increases, which aligns with the existing
literature \citep{fox1994birth,walker2007teen}. The \Noweights{}(LM) estimator shows a linearly decreasing pattern, while the \Noweights{}(KRR) estimator stops decreasing after age 30. For three weighted estimators, the effects are stable around age 27 to 32, but tend to decrease quickly after age 32.   Compared to IPW(LM) and IPW(KRR), Proposed does not show the increasing tendency before age 30 and the decrease after age 32 is relatively smoother.

\section{Discussions}
The PCATE characterizes subgroup treatment effects and provides insights about how treatment effect varies across the characteristics of interest.
We develop a novel nonparametric estimator for the PCATE
under treatment ignorability. 
The proposed hybrid kernel weighting is a non-trivial extension of covariate balancing weighting in the ATE estimation literature in that it aims to achieve approximate covariate balancing for all flexible outcome mean functions and for all subgroups defined based on continuous variables. In contrast to existing estimators, we do not require any smoothness assumption on the propensity score, and thus our weighting approach is particularly useful in studies when the treatment assignment mechanism is quite complex. 

We conclude with 
several interesting and important extensions of the current estimator 
as future research directions. 
First, an improved data-adaptive  bandwidth selection procedure is worth investigating as it plays an important role in smoothing. In addition, instead of local constant regression,  other alternatives such as linear or spline smoothers  can be considered. 
Third, given the appealing theoretical properties, we will investigate efficient computation of the proposed weighting estimators with $L_\infty$-norm. 
Furthermore, the asymptotic distribution of proposed estimator is worth studying so that inference procedures can be developed. 

\section*{Acknowledgement}
The work of Raymond K. W. Wong is partially supported by the National Science Foundation (DMS-1711952 and CCF-1934904). The work of Shu Yang is partially supported by the National Institute on Aging (1R01AG066883) and the National Science Foundation (DMS-1811245). The work of Kwun Chuen Gary Chan is partially supported by the National Heart, Lung, and Blood Institute (R01HL122212) and the National Science Foundation (DMS-1711952).

\appendix

\bibliographystyle{chicago}
\bibliography{refer2.bib}

\makeatletter\@input{xx1.tex}\makeatother
\end{document}


\def\spacingset#1{\renewcommand{\baselinestretch}%
{#1}\small\normalsize}

\linsps

  \title{\bf Supplementary material for ``Estimation of Partially Conditional Average Treatment Effect by Hybrid Kernel-covariate Balancing''}
	
	\author{
		Jiayi Wang\thanks{
		Jiayi Wang is Ph.D. candidate, Department of Statistics, Texas A\&M University, College Station, TX 77843, USA (Email: {\tt jiayiwang@stat.tamu.edu}).}\,,
		Raymond K. W. Wong\thanks{
		Raymond K. W. Wong is Associate Professor, Department of Statistics, Texas A\&M University,
			College Station, TX 77843, USA (Email: {\tt raywong@tamu.edu}).
			His research is partially supported by the National Science Foundation (DMS-1711952 and CCF-1934904).
		}\,,
		Shu Yang\thanks{
		Shu Yang is Assistant Professor, Department of Statistics, North Carolina State University, Raleigh, NC 27695, USA (Email: {\tt syang24@ncsu.edu}).
		Her research is partially supported by the National Institute on Aging (1R01AG066883)
		and the National Science Foundation (DMS-1811245).
		}\,
		and
		Kwun Chuen Gary Chan\thanks{
		Kwun Chuen Gary Chan is Professor, Department of Biostatistics, University of Washington, Seattle, WA 98195, USA (Email: {\tt kcgchan@u.washington.edu}).
		His research is partially supported by
		the National Heart, Lung, and Blood Institute (R01HL122212)
		and the National Science Foundation (DMS-1711952).
		}}
		\date{}
  \maketitle

\spacingset{1}
\section{Computation}
\label{sec:comp_supp}
\subsection{Reparametrization}
\label{sec:reparameterization}
To solve \eqref{eqn:obj_final},  we focus on the inner optimization of \eqref{eqn:obj_all}:
$\sup_{u\in\Hscr_{N}}\{S_{N,h}(w,u)-\lambda_{1}\|u\|_{\Hscr}^{2}\}$, which is equivalent to
\begin{align}
\label{eqn:inner_eq1}
    \sup_{u\in\Hscr}\left\{\frac{S_{N,h}(w,u)}{\|u\|_N}-\lambda_{1}\frac{\|u\|_{\Hscr}^{2}}{\|u\|_N}\right\}.
\end{align}

By the representer theorem \citep{wahba1990spline}, the solution
of this infinite dimensional optimization \eqref{eqn:inner_eq1} can be shown to lie in a
finite dimensional subspace of $\Hscr$: $\mathrm{span}\{\kernel(X_{i},\cdot):i=1,\dots,N\}$.
Take $M=[\kernel(X_{i},X_{j})]_{i,j=1}^{N}\in\mathbb{R}^{N\times N}$,
we obtain
\begin{align}
\sup_{u\in\Hscr_{N}}\left\lbrace S_{N,h}(w,u)-\lambda_{1}\|u\|_{\Hscr}^{2}\right\rbrace =\sup_{\alpha\in\mathbb{R}^{N}}\left[\frac{S_{N,h}\{w,\sum_{j=1}^{N}\alpha_{j}\kernel(X_{j},\cdot)\}}{\alpha^{\tp}M^{2}\alpha/N}-\lambda_{1}\frac{\alpha^{\tp}M\alpha}{\alpha^{\tp}M^{2}\alpha/N}\right].\label{eqn:inner}
\end{align}
By the definition of $S_{N,h}(w,u)$ in \eqref{eqn:S_form}, we have
\begin{align*}
S_{N,h}\left\{ w,\sum_{j=1}^{N}\alpha_{j}\kernel(X_{j},\cdot)\right\} =\frac{1}{N^{2}}\alpha^{\tp}M\mathrm{diag}(T\circ w-J)G_{h}\mathrm{diag}(T\circ w-J)M\alpha,
\end{align*}
where $\circ$ represents the element-wise
product of two vectors, $J = [1,1,\dots,1]\in \mathbb{R}^{N}$, and 
\[
G_{h}=\left[\begin{array}{ccc}
\int_{\mathcal{V}}\tilde{K}_{h}(V_{1},v)\tilde{K}_{h}(V_{1},v)\de v & \cdots & \int_{\mathcal{V}}\tilde{K}_{h}(V_{1},v)\tilde{K}_{h}(V_{N},v)\de v\\
\vdots & \ddots & \vdots\\
\int_{\mathcal{V}}\tilde{K}_{h}(V_{N},v)\tilde{K}_{h}(V_{1},v)\de v & \cdots & \int_{\mathcal{V}}\tilde{K}_{h}(V_{N},v)\tilde{K}_{h}(V_{N},v)\de v
\end{array}\right].
\]
Note that $M$ is positive semi-definite.
We consider the eigen-decomposition of $M$ as
\begin{align}
M=PDP^{\tp}\label{eqn:M_eigen}
\end{align}
where $D\in\mathbb{R}^{r\times r}$ is a diagonal matrices
with nonzero diagonal entries, and $P\in\mathbb{R}^{N\times r}$ is
an orthonormal matrix. Take $\beta=N^{-1/2}DP^{\tp}\alpha$. Then
\eqref{eqn:inner} is equivalent to 
\begin{align*}
\sup_{\beta\in\mathbb{R}^{r}:\|\beta\|_{2}\leq1}\beta^{\tp}\left\lbrace \frac{1}{N}P^{\tp}\mathrm{diag}(T\circ w-J)G_{h}\mathrm{diag}(T\circ w-J)P-N\lambda_{1}D^{-1}\right\rbrace \beta.
\end{align*}
Therefore,
\begin{align}
\hat{w}=\argmin_{w\ge1}\left[\sigma_{1}\left\lbrace \frac{1}{N}P^{\tp}\mathrm{diag}(T\circ w-J)G_{h}\mathrm{diag}(T\circ w-J)P-N\lambda_{1}D^{-1}\right\rbrace +\lambda_{2}R_{N,h}(w)\right].
\end{align}

\subsection{Proof of Lemma \ref{lemma:convex}}
\label{sec:proof_convex}
\begin{proof}[Proof of Lemma \ref{lemma:convex}]
By the definition \eqref{eqn:obj_V}, $R_{N,h}(w)$ is a convex function of
$w$. Also, $P^{\tp}(T\circ w-J)$ is an affine transformation of $w$.
Then it suffices to show that $\sigma_{1}\{\mathrm{diag}(y)G_{h}\mathrm{diag}(y)+B\}$
is a convex function of $y$  for any symmetric matrix $B\in\mathbb{R}^{r\times r}$.

First, we show that $G_{h}$ is a positive semi-definite matrix.
For any vector $a\in\mathbb{R}^{N}$, 
\begin{align*}
a^{\tp}G_{h}a & =\int_{\mathcal{V}}a^{\tp}\left[{\begin{array}{ccc}
\tilde{K}_{h}(V_{1},v)\tilde{K}_{h}(V_{1},v) & \cdots & \tilde{K}_{h}(V_{1},v)\tilde{K}_{h}(V_{N},v)\\
\vdots & \ddots & \vdots\\
\tilde{K}_{h}(V_{N},v)\tilde{K}_{h}(V_{1},v) & \cdots & \tilde{K}_{h}(V_{N},v)\tilde{K}_{h}(V_{N},v)
\end{array}}\right]a\, dv\\
 &= \int_{\mathcal{V}} \left(\sum_{j=1}^N \tilde{K}_{h}(V_{j},v) a_j\right)^2 dv \ge 0
\end{align*}
Therefore there exists a matrix $L$ such that $G_{h}=LL^{\tp}$.

Consider any vector $y_{1},y_{2}\in\mathbb{R}^{r}$, and $t\in[0,1]$.
For $\beta\in\mathbb{R}^{r}$, 
\begin{align*}
 & \beta^{\tp}\left[\mathrm{diag}\{ty_{1}+(1-t)y_{2}\}G_{h}\mathrm{diag}\{ty_{1}+(1-t)y_{2}\}+B\right]\beta\\
&= \beta^{\tp}\left[\mathrm{diag}\{ty_{1}+(1-t)y_{2}\}LL^{\tp}\mathrm{diag}\{ty_{1}+(1-t)y_{2}\}+B\right]\beta\\
&= \left\Vert L^{\tp}\mathrm{diag}\{ty_{1}+(1-t)y_{2}\}\beta\right\Vert _{2}^{2}+\beta^{\tp}B\beta\\
&= \left\Vert tL^{\tp}\mathrm{diag}(y_{1})\beta+(1-t)L^{\tp}\mathrm{diag}(y_{2})\beta\right\Vert _{2}^{2}+\beta^{\tp}B\beta\\
&\leq t\|L^{\tp}\mathrm{diag}(y_{1})\beta\|_{2}^{2}+(1-t)\|L^{\tp}\mathrm{diag}(y_{2})\beta\|_{2}^{2}+\beta^{\tp}B\beta\\
&= t\beta^{\tp}\{\mathrm{diag}(y_{1})G_{h}\mathrm{diag}(y_{1})+B\}\beta+(1-t)\beta^{\tp}\{\mathrm{diag}(y_{2})G_{h}\mathrm{diag}(y_{1})+B\}\beta,
\end{align*}
where the above inequality is due to the fact that $\|y\|_{2}^{2}$ is a convex
function of $y$. Therefore, we have 
\begin{multline*}
\sigma_{1}\left(\mathrm{diag}\{ty_{1}+(1-t)y_{2}\}G_{h}\mathrm{diag}\{ty_{1}+(1-t)y_{2}\}+B\right)\\
\leq t\sigma_{1}\left(\mathrm{diag}(y_{1})G_{h}\mathrm{diag}(y_{1})+B\right)+(1-t)\sigma_{1}\left(\mathrm{diag}(y_{2})G_{h}\mathrm{diag}(y_{2})+B\right),
\end{multline*}
which leads to the conclusion.
\end{proof}

\section{Proofs of Theorems}
Through out the proof, we use $\check{x}$ to represent  a  generic  vector  in $\mathcal{X}$, and use $\check{v} \in \mathcal{V}$ to represent the sub-vector of $\check{x}$ that is of interest.
\subsection{Proof of Theorem \ref{thm: prep}}
\label{sec:proof_prep}
Firstly, we introduce some notations. Take  $\gamma_i := T_i w_i^*-1$, where $w_i^* = 1/\pi(X_i)$ for $i = 1,\dots, N$. And define  $\Hscr(1)  := \{ u \in \Hscr: \|u\|_\Hscr \leq 1\}$.
Due to Lemma 2.1 of \cite{lin2000tensor}, there exists a constant $b$ such that $\sup_{u \in \Hscr(1)} |u|_\infty \leq b$.

We replace $\frac{1}{Nh^{d_1}} \sum_{j=1}^N K(\frac{V_j - v}{h})$ in $S_{N,h}(w^*, u)$
with its expectation $g_h(v)$
and obtain
\begin{align}
	\label{eqn:Z_exp}
	\tilde S_{N,h}(w^*,u) & := \left\| \frac{1}{g_h(\cdot)}\left\lbrace \frac{1}{Nh^{d_1}}\sum_{i=1}^{N} \gamma_iu(X_i) K\left(\frac{V_i-\cdot}{h}\right) \right\rbrace\right\|_2^2.
\end{align}
Next, we show that $g_h$ is lower bounded. Based on Assumption \ref{assum: ban}, without loss of  generality, we take $h \leq 1$ .
Under Assumption \ref{assum: density}, there exists a constant $\newl \label {g_lower_const}$ such that
\begin{align}
	\begin{split}
		g_h(v) & =\E \frac{1}{h^{d_1}} K\left(\frac{V_i- v}{h}\right) =  \frac{1}{h^{d_1}} \int_{ I}  K\left(\frac{V- v}{h}\right) g(V) d V = \int_{(zh+v)\in[0,1]^{d_1}} K(z) g(zh+v)dz                                                                \\
		& \ge \oldu{density_lower} \int_{(zh+v)\in[0,1]^{d_1}} K(z) dz \ge  \oldu{density_lower} \int_{(z+v)\in[0,1]^{d_1}} K(z) dz \\
		&\ge  \oldu{density_lower} \min \left\{ \int_{[0,1/2]^{d_1}} K(z) dz , \int_{[-1/2,0]^{d_1}} K(z) dz\right\} \ge \oldl{g_lower_const}. \label{eqn:g_bound}
	\end{split}
\end{align}
Then,
\begin{align}
	\begin{split}
		\tilde S_{N,h}(w^*,u) & \leq  \frac{1}{\inf_{v\in [0,1]^{d_1} }g^2_h(v)} \frac{1}{h^{2d_1}} \left\| \frac{1}{N}\sum_{i=1}^{N} \gamma_iu(X_i) K\left(\frac{V_i-\cdot}{h}\right) \right\|_2^2 \\
		& \leq  \frac{1}{\oldl{g_lower_const}^2 h^{2d_1}} \left\| \frac{1}{N}\sum_{i=1}^{N} \gamma_iu(X_i) K\left(\frac{V_i-\cdot}{h}\right) \right\|_2^2.
	\end{split}
	\label{eqn:tilde_S}
\end{align}

Below, we will establish the bound of $\left| \frac{1}{N}\sum_{i=1}^{N}
	\gamma_iu(X_i) K\left((V_i-v)/{h}\right)\right|$ uniformly for every $u \in \Hscr_N$
by conditioning on $v$.
To start with, we define  $\|f\|_N := \sqrt{\frac{1}{N} \sum_{i=1}^{N} f^2(X_i)}$ for some function $f$,
$\Kscr_h : = \left\lbrace  K\left((\cdot-v)/{h}\right) : v\in[0,1]^{d_1}  \right\rbrace $, $\knnorm := \sup_{\tilde{v} \in [0,1]^{d_1}} \sqrt{\frac{1}{N} \sum_{i=1}^{N} K^2((V_i - \tilde{v})/h)}$.
And we take 
$\Fscr_{h,v} := \left\{  f: f(\check{x}) = u(\check{x})K(\frac{\check{v} - v}{h}); u \in \Hscr(1) \right\}$.

The next lemma provides an entropy bound for the space $\Fscr_{h,v}$.

\begin{lemma}
	\label{entropy}
	For every fixed $h$ and $v$,  there exists a constant $A>0$, 
	such that
	\begin{equation*}
		H(\delta, \Fscr_{h,v}, \|\cdot\|_N)
		\begin{cases}
			= 0                                   & \mbox{if $\delta>2b\knnorm$} \\
			\le A\knnorm^{\alpha}\delta^{-\alpha} & \mbox{otherwise}
		\end{cases}.
	\end{equation*}
\end{lemma}

\begin{proof}
	Notice that for every $f_1$, $f_2\in \Fscr_{h,v}$, $\|f_1-f_2\|_{N} \leq \|f_1\|_N+ \|f_2\|_N \leq 2b\knnorm$. Therefore, 	$H(\delta, \Fscr_{h,v}, L^2(\pr_N)) = 0,   \mathrm{when} \  \delta> 2b\knnorm$.\\
	By \citet{birman1967piecewise},
	we have $H(\epsilon, \Hscr(1), \|\cdot\|_\infty)\leq A\epsilon^{-\alpha}$ for some constant $A>0$. Therefore,  the covering number $\mathcal{N}(\epsilon, \Hscr(1), \|\cdot\|_\infty) \leq \exp(A\epsilon^{-\alpha})$.

	Take $\mathcal{N} \subset \Hscr(1)$  as the $\epsilon$-net of $\Hscr(1)$
	with respect to $\|\cdot\|_\infty$.
	By definition, for every $u \in \Hscr(1)$,  there exists a $u_0\in \mathcal{N}$, such that
	\begin{align}
		\label{eqn:entropy_u}
		\sup_{x \in [0,1]^d} |u(x) - u_0(x)| \leq \epsilon.
	\end{align}

	Take 
	$\mathcal{N}_v := \left\lbrace f: f(\check{x}) = u(\check{x})K((\check{v} - v)/{h}); u\in \mathcal{N}   \right\rbrace $.
	Then, for every $f \in \Fscr_{h,v}$, there exists a $f_0 \in \mathcal{N}_v$, such that
	\begin{align*}
		\|f - f_0\|_N^2 & = \frac{1}{N} \sum_{i = 1}^N \left| u(X_i)K\left(\frac{V_i - v}{h}\right) - u_0(X_i)K\left(\frac{V_i - v}{h}\right)\right|^2 \\
		                & = \frac{1}{N} \sum_{i = 1}^N K^2\left(\frac{V_i - v}{h}\right)\left| u(X_i)- u_0(X_i)\right|^2                               \\
		                & \leq  \sup_{x\in [0,1]^d}\left| u(x)- u_0(x)\right|^2 \frac{1}{N} \sum_{i = 1}^N K^2\left(\frac{V_i - v}{h}\right)           \\
		                & \leq  \epsilon^2 \knnorm^2.
	\end{align*}
	The last inequality is due to \eqref{eqn:entropy_u} and $ \frac{1}{N} \sum_{i = 1}^N K^2\left((V_i - v)/{h}\right) \leq \knnorm^2$.
	Therefore, we have
	\begin{align*}
		\mathcal{N}(\epsilon \knnorm, \Fscr_{h,v}, \|\cdot\|_N)  \leq \mathcal{N}(\epsilon, \Hscr(1), \|\cdot\|_\infty)  \leq \exp\left(A\epsilon^{-\alpha}\right).
	\end{align*}
	The conclusion follows by taking $\delta = \epsilon\knnorm$.
\end{proof}

Then, we study the concentration property of the terms $\knnorm$ and $\sum_{i=1}^N K((V_i - \tilde v)/h))/(Nh^{d_1})$.
\begin{lemma}
	\label{norm_control}
	Under Assumptions \ref{assum: kernel_2}, \ref{assum: density} and \ref{assum: ban},	there exist constants $\newl\ltxlabel{kp1}, \newl\ltxlabel{kernel_exp},\newl\ltxlabel{kp2} >0$ depending  on $\oldu{supkernel}$, $\oldu{density_lower}$, $A_1$ and $\nu_1$,
	such that, for all sufficiently large $N$, the following hold:
	\begin{align}
		\E \knnorm^2                                         & \leq \oldl{kernel_exp}h^{d_1},
		\label{eqn:exp_control}                                                                                                         \\
		\pr(\knnorm^2  \ge  2\ct \oldl{kernel_exp} h^{d_1} ) & <c\exp\left\lbrace - \oldl{kp1}\ct Nh^{d_1} \right\rbrace, \quad t\ge 1,
		\label{eqn:norm_control}                                                                                                        \\
		\pr\left(    \sup_{\tilde{v} \in [0,1]^{d_1}} \left|\frac{1}{Nh^{d_1}} \sum_{i=1}^N K\left( \frac{{V_i }- \tilde{v}}{h}  \right) - g_h(\tilde{v})\right|  \ge   \ct \oldl{g_lower_const} \right)
		                                                     & \leq c\exp\left\lbrace  -\oldl{kp2}\ct Nh^{d_1}\right\rbrace,
		\quad \frac{1}{2}\le t <1.
		\label{eqn:denorm_control}
	\end{align}
\end{lemma}

\begin{proof}

	Take $r_i$, $i = 1,\dots,n$, as independent Rademacher random variables. We have
	\begin{align*}
		\E    \knnorm^2
		 & = \E\sup_{{v} \in [0,1]^{d_1}} \frac{1}{N} \sum_{i=1}^N K^2\left( \frac{{V_i }- \tilde{v}}{h}  \right)                                                                                                                                                                 \\
		 & \leq \E\sup_{\tilde{v}\in [0,1]^{d_1}}\E K^2\left( \frac{{V_i }- \tilde{v}}{h}  \right) +	\E \sup_{\tilde{v} \in [0,1]^{d_1}} \left|\frac{1}{N} \sum_{i=1}^N K^2\left( \frac{{V_i }- \tilde{v}}{h}  \right) - \E K^2\left( \frac{{V_i} - \tilde{v}}{h}  \right) \right| \\
		 & = \sup_{\tilde {v}\in [0,1]^{d_1}}  \E K^2\left( \frac{{V_i }- \tilde{v}}{h}  \right) +	\E \sup_{\tilde{v} \in [0,1]^{d_1}} \left|\frac{1}{N} \sum_{i=1}^N K^2\left( \frac{{V_i }- \tilde{v}}{h}  \right) - \E K^2\left( \frac{{V_i} - \tilde{v}}{h}  \right) \right|   \\
		 & \leq \sup_{\tilde{v}\in [0,1]^{d_1}} \E K^2\left( \frac{{V_i} - \tilde{v}}{h}  \right) + 2\E \sup_{\tilde{v} \in [0,1]^{d_1}} \left|\frac{1}{N} \sum_{i=1}^N r_i K^2\left( \frac{{V_i }- \tilde{v}}{h}  \right) \right|                                                \\
		 & \leq  \sup_{\tilde{v}\in [0,1]^{d_1}} \E K^2\left( \frac{{V_i }- \tilde{v}}{h}  \right) + 8\oldu{supkernel} \E \sup_{\tilde{v} \in [0,1]^{d_1}} \left|\frac{1}{N} \sum_{i=1}^N r_i  K\left( \frac{{V_i }- \tilde{v}}{h}  \right)  \right|.
	\end{align*}
	The second last inequality is due to the symmetrization inequality  from Theorem 2.1 in  \citet{koltchinskii2011oracle}, while the last inequality is due to the contraction inequality from Theorem 2.3 in  \citet{koltchinskii2011oracle}.  Next,  we bound the Rademacher  complexity $$\E\|R_N\|_{\Kscr_h }: = \E \sup_{\tilde{v} \in [0,1]^{d_1}} \left|\frac{1}{N} \sum_{i=1}^N r_i  K\left( \frac{V_i - \tilde{v}}{h}  \right)  \right|.$$

	Since $\Kscr_h \subset \Kscr$, from the  entropy bound in Assumption \ref{assum: kernel_2} for $\Kscr$,  we have $\mathcal{N}(\varepsilon, \Kscr_h, \|\cdot\|_{N})\leq A_1 \varepsilon^{-\nu_1}$. Define $\sigma^2_{\Kscr_h} :=   \sup_{\tilde{v} \in [0,1]^{d_1}} \E K^2((V_i- \tilde{v})/h)$.
	By applying Theorem 3.12 in \citet{koltchinskii2011oracle}, we have
	\begin{align}
		\label{eqn:rade}
		\E\|R_N\|_{\Kscr_h }  \leq c \left[  \sqrt{\frac{\nu_1}{N}}\sigma_{\Kscr_h} \sqrt{\log\frac{A_1\oldu{supkernel} }{\sigma_{\Kscr_h} }} + \frac{\nu_1 \oldu{supkernel}}{N} \log\frac{A_1\oldu{supkernel}}{\sigma_{\Kscr_h} } \right],
	\end{align}
	where $c>0$ is an universal constant. Next,
	\begin{align}
		 & \sigma_{\Kscr_h}^2 = \sup_{\tilde{v} \in [0,1]^{d_1}} \int_0^1  K^2 \left(\frac{v - \tilde{v}}{h}\right) g(v) dv = h^{d_1} \sup_{\tilde{v} \in [0,1]^{d_1}}\int_{(zh+\tilde{v})\in [0,1]^{d_1}} K^2(z) g(zh+ \tilde{v})dz, \\
		 & \oldu{density_lower} h^{d_1} \sup_{\tilde{v} \in [0,1]^{d_1}}\int_{(zh+\tilde{v})\in [0,1]^{d_1}} K^2(z) dz\leq  \sigma_{\Kscr_h}^2 \leq \oldu{supkernel}^2 h^{d_1}, \label{eqn:sigma_kh1}                                 \\
		 & \oldu{density_lower} h^{d_1} \int_{[0,1]^{d_1}} K^2(z) dz\leq  \sigma_{\Kscr_h}^2 \leq \oldu{supkernel}^2 h^{d_1} \label{eqn:sigma_kh2},
	\end{align}
	where \eqref{eqn:sigma_kh1} is due to  $g(\cdot) \ge  \oldu{density_lower}$ and $K(\cdot) \leq \oldu{supkernel}$; \eqref{eqn:sigma_kh2} is valid for $h\leq 1$. Since   $\int_{[0,1]^{d_1}} K^2(z) dz>0$,   we have $\sigma_{\Kscr_h}^2 \asymp h^{d_1}$.

	Therefore, there exists a constant $\oldl{kernel_exp}>0$ depending on $\oldu{supkernel}$, $\oldu{density_lower}$, $\nu_1$ and $A_1$, such that
	\begin{align*}
		\E    \knnorm^2 & \leq \sigma^2_{\Kscr_h} + 8\oldu{supkernel} \E\|R_N\|_{\Kscr_h }                                                                                                                                                                                      \\
		                & \leq  \sigma^2_{\Kscr_h}  + 8\oldu{supkernel} c \left[  \sqrt{\frac{\nu_1}{N}}\sigma_{\Kscr_h} \sqrt{\log\frac{A_1\oldu{supkernel} }{\sigma_{\Kscr_h} }} + \frac{\nu_1 \oldu{supkernel}}{N} \log\frac{A_1\oldu{supkernel}}{\sigma_{\Kscr_h} } \right] \\
		                & \leq \oldl{kernel_exp} h^{d_1}
	\end{align*}
	The last inequality is due to Assumption \ref{assum: ban} and it is valid for all large enough $N$.

	From Talagrand's inequality (Theorem 2.5 in \citet{koltchinskii2011oracle}),  and $$\sup_{\tilde{v}\in[0,1]^{d_1}}\E K^4\left(\frac{V-\tilde{v}}{h}\right)\leq \oldu{supkernel}^4h^{d_1},$$
	we have for any $\ct \ge 1$,
	\begin{align*}
		\pr(\knnorm^2  \ge  2\ct \oldl{kernel_exp} h^{d_1} ) 
		\leq  c\exp\left\lbrace -\oldl{kp1}\ct Nh^{d_1} \right\rbrace,
	\end{align*}
	where $c>0$ is an universal constant and $\oldl{kp1}>0$ is a constant depending on $\oldu{supkernel}$, $\oldu{density_lower}$, $\nu_1$ and $A_1$.

	Also, by adopting symmetrization inequality again,  there exists a constant  $\newl\label{cd1}>0$ depending on $A_1$, $\nu_1$ and $\oldu{supkernel}$ such that
	\begin{align}
		 & \E  \sup_{\tilde{v} \in [0,1]^{d_1}} \left|\frac{1}{N} \sum_{i=1}^N K\left( \frac{{V_i }- \tilde{v}}{h}  \right) - \E K\left( \frac{{V_i} - \tilde{v}}{h}  \right)  \right|  \leq 2\E \|R_N\|_{\Kscr_h}   \nonumber                                                              \\
		 & \leq 2 c \left[  \sqrt{\frac{\nu_1}{N}}\sigma_{\Kscr_h} \sqrt{\log\frac{A_1\oldu{supkernel} }{\sigma_{\Kscr_h} }} + \frac{\nu_1 \oldu{supkernel}}{N} \log\frac{A_1\oldu{supkernel}}{\sigma_{\Kscr_h} } \right] \label{eqn:eq1}\\
		 & \leq \oldl{cd1} N^{-1/2} h^{d_1/2}\sqrt{\log 1/h^{d_1}}, \nonumber
	\end{align}
	where the last inequality is due to Assumption \ref{assum: ban}, and the first term of \eqref{eqn:eq1} is dominant for large enough $N$.

	By Talagrand's inequality, for any $\ct > 0 $, we have
	\begin{align*}
		 & \pr\left(    \sup_{\tilde{v} \in [0,1]^{d_1}} \left|\frac{1}{N} \sum_{i=1}^N K\left( \frac{{V_i }- \tilde{v}}{h}  \right) - \E K\left( \frac{{V_i }- \tilde{v}}{h}  \right)  \right|  \ge  \oldl{cd1} N^{-1/2} h^{d_1/2}\sqrt{\log 1/h^{d_1}}  + \ct \right) \\
		 & \leq c \exp\left\lbrace   -\frac{1}{c} \frac{N^2\ct^2}{\tilde{V} + n\ct\oldu{supkernel}}\right\rbrace,                                                                            
	\end{align*}
	where $\tilde{V}: = N\oldu{supkernel}^2 h^{d_1}+ 16\oldu{supkernel}\oldl{cd1} N^{1/2} h^{d_1/2} \sqrt{\log  1/h^{d_1} }
		\leq 2 N\oldu{supkernel}^2 h^{d_1} $,  for all large enough $N$.

	Take $\ct = \ct' \oldl{g_lower_const} h^{d_1} - \oldl{cd1} N^{-1/2} h^{d_1/2} \sqrt{\log  1/h^{d_1} } $, for  $ 1/2 \leq  \ct' < 1$. 
	For all large enough $N$, we have $\ct \ge \ct'\oldl{g_lower_const} h^{d_1}/2$.   Therefore, we have
	\begin{align*}
		 & \pr\left(    \sup_{\tilde{v} \in [0,1]^{d_1}} \left|\frac{1}{Nh^{d_1}} \sum_{i=1}^N K\left( \frac{V_i - \tilde{v}}{h}  \right) - g_h(\tilde{v})\right|  \ge   \ct' \oldl{g_lower_const}\right)                                  \\
		 & =\pr\left(    \sup_{\tilde{v} \in [0,1]^{d_1}} \left|\frac{1}{N} \sum_{i=1}^N K\left( \frac{V_i - \tilde{v}}{h}  \right) - \E K\left( \frac{V_i - \tilde{v}}{h}  \right)  \right|  \ge   \ct'\oldl{g_lower_const}h^{d_1}\right) \\
		 & \leq c\exp\left\lbrace  -\oldl{kp2} \ct' Nh^{d_1}\right\rbrace ,
	\end{align*}
	where $c>0$ is universal constant and $\oldl{kp2}>0$ is a constant depending on $\oldu{supkernel}, \oldu{density_lower}$, $A_1$ and $\nu_1$.

\end{proof}

Next, we derive the bound for $|\sum_{i=1}^N \gamma_i f(X_i)/N|$ uniformly for every $f \in \Fscr_{h,v}$.
\begin{lemma}
	\label{Z_term}
	Under Assumptions \ref{assum: prop}-\ref{assum: ban},	there exists constants $\newl\ltxlabel{cfpf}, \newl\ltxlabel{cfpf2}>0$ depending on $b,\oldu{supkernel}, A,\oldu{prop}$
	and $\alpha$ such that
	with probability at least  $1- c\exp\left(-{\oldl{cfpf}}\ct  \right)$,
	$$
		\forall f \in \Fscr_{h,v}, \qquad \frac{1}{N} \left| \sum_{i=1}^N {\gamma}_i f({X_i}) \right| \leq \ct \left\{ N^{-\frac{1}{2}}\|u\|_2^{\frac{2-\alpha}{2p}} h^{d_1\left(\frac{1}{2} - \frac{2-\alpha}{4p} \right)} + N^{-\frac{2}{2+\alpha}} h^{\frac{d_1\alpha}{2+\alpha}} \right\},
	$$
	for any $\ct \ge \oldl{cfpf2}$ and $p \ge 1$.
\end{lemma}

\begin{proof}

	Since	$\E(\tilde{\gamma}_i \mid {X_i}) = 0$ and ${\gamma}_i\mid {X_i} $, $i=1,\dots, n$, are bounded sub-gaussian random variables.
	Therefore, there exists a constant $\sigma_\gamma>0$ depending on $\oldu{prop}$,   such that
	$\E\left\lbrace \exp(\lambda \gamma )|X=x \right\rbrace \leq \exp (\lambda^2\sigma_\gamma^2/2)$ for every $x$.

	Define $\mathcal{F}_{h,v}(\delta):=\{f\in\mathcal{F}_{h,v}:\|f\|_2\leq \delta\}$ for $\delta > 0$ .
	We begin by deriving an upper bound for
	$\E[\sup_{f\in\mathcal{F}_{h,v}(\delta)} \sum^N_{i=1} {\gamma}_i f({X}_i)/N]$.
	Conditioned on ${X_i}, i = 1,\dots,N$,
	$
		\sum_{i=1}^{N}{ \gamma}_if({X_i})/\sqrt{N}
	$
	is  a sub-gaussian  process with respect to the metric space $(\Fscr_{h,v}, \mathrm{dist})$, where $\mathrm{dist}^2(f_1, f_2) = \frac{\sigma_\gamma^2}{N} \sum_{i=1}^{N} (f_1({X_i}) - f_2({X_i}))^2$ for $f_1, f_2 \in \Fscr_{h,v}$.
	Therefore, by Dudley's entropy bound, and Lemma \ref{entropy},   for any $\delta>0$,
	we have
	\begin{align*}
		\E \left\{ \sup_{f \in \Fscr_{h,v}(\delta)} \frac{1}{\sqrt{N}} \left| \sum_{i=1}^N { \gamma}_i f({X_i})  \right| \mid {X_i}, i = 1,\dots, N \right\} \leq c \int_{0}^{2\sigma_\gamma \delta_{N}} \sqrt{H(\tau, \Fscr_{h,v}, \|\cdot\|_N)} d\tau,
	\end{align*}
	where $\delta_{N}^2 = \sup_{f \in \Fscr_{h,v} (\delta)} \left|   \frac{1}{N} \sum_{i=1}^N f^2({X_i}) \right|$.

	Taking expectations on both sides and using Lemma \ref{entropy}, there exists a constant $\newl\ltxlabel{cf_exp}>0$ depending on $A$, $\sigma_\gamma$,  $\alpha$ and $\oldl{kernel_exp}$  such that
	\begin{align*}
		\E  \sup_{f \in \Fscr_{h,v}(\delta)} \frac{1}{N} \left| \sum_{i=1}^N { \gamma}_i f({X_i}) \right|
		 & \leq  \frac{c}{\sqrt{N}} \E \int_{0}^{2\sigma_\gamma \delta_N} \sqrt{H(\tau, \Fscr,\|\cdot\|_N)}d\tau                                                                   \\
		 & \leq \frac{c}{\sqrt{N}}  \E \int_{0}^{2\sigma_\gamma \delta_{N}} A^{1/2}\knnorm^{\alpha/2} \tau^{-\alpha/2}d\tau                                                        \\
		 & \leq \frac{cA^{1/2}}{\sqrt{N}} \frac{1}{1-\alpha/2} \E \knnorm^{\alpha/2} (2\sigma_\gamma\delta_N)^{1-\alpha/2}                                                         \\
		 & =  \frac{cA^{1/2}}{\sqrt{N}} \frac{(2\sigma_\gamma)^{1-\alpha/2}}{1-\alpha/2}  \E \knnorm^{\alpha/2} \delta_N^{1-\alpha/2} \qquad \mbox{(by H\"older's Inequality)}        \\
		 & \leq  \frac{cA^{1/2}}{\sqrt{N}} \frac{(2\sigma_\gamma)^{1-\alpha/2}}{1-\alpha/2}  (\E\delta_N)^{1-\alpha/2} (\E\knnorm)^{\alpha/2} \qquad \mbox{(by Jensen's Inequality)} \\
		 & \le  \frac{cA^{1/2}}{\sqrt{N}} \frac{(2\sigma_\gamma)^{1-\alpha/2}}{1-\alpha/2}   (\E\delta_N^2)^{\frac{1-\alpha/2}{2}} (\E\knnorm^2)^{\alpha/4}
		\qquad \mbox{(by \eqref{eqn:exp_control} in Lemma \ref{norm_control})}                                                                                                       \\
		 & \le \frac{cA^{1/2}}{\sqrt{N}} \frac{(2\sigma_\gamma)^{1-\alpha/2}}{1-\alpha/2}   (\E\delta_N^2)^{\frac{1-\alpha/2}{2}} (\oldl{kernel_exp}h^{d_1})^{\alpha/4}            \\
		 & \leq  \oldl{cf_exp}N^{-1/2} h^{d_1\alpha/4}  (\E\delta_N^2)^{\frac{1-\alpha/2}{2}}
	\end{align*}

	Next, we derive an upper bound for $\E\delta_N^2$.
	By symmetrization and contraction inequalities,
	\begin{align*}
		\E \delta_{N}^2 & \leq \delta ^2 + 2\E  \sup_{f\in \Fscr_{h,v}(\delta)} \left| \frac{1}{N} \sum_{i=1}^{N} f^2({X_i})- \E f^2({X_i})\right|       \\
		                & \leq \delta ^2 + 2\E  \sup_{f\in \Fscr_{h,v}(\delta)}\left|  \frac{1}{N} \sum_{i=1}^{N} r_i f^2({X_i}) \right|                 \\
		                & \leq \delta^2 + 8 b \oldu{supkernel}\E  \sup_{f\in \Fscr_{h,v}(\delta)} \left| \frac{1}{N} \sum_{i=1}^{N} r_i f({X_i})\right|,
	\end{align*}
	where  $r_i$, $i=1,\dots,n$, are independent Rademacher random
	variables.
	Applying the entropy bound from Lemma \ref{entropy} and with Theorem 3.12 in 
	, we have
	\begin{align*}
		\E  \sup_{f\in \Fscr_{h,v}(\delta)} \left| \frac{1}{N} \sum_{i=1}^{N} r_i f({X_i})\right|  \leq \newl\ltxlabel{rade_const} \max \left\{
		\frac{h^{d_1\alpha/4}}{\sqrt{N}} \delta^{1-\alpha/2} , \frac{h^{d_1\alpha/(2+\alpha)}}{N^{2/(2+\alpha)}}
		\right\}
	\end{align*}
	for some constant $\oldl{rade_const}>0$ depending on $A, b, \oldu{supkernel}, \alpha$.

	We now combine the above results.
	Also, as Assumption \ref{assum: ban} indicates,  for some constants $\newl\label{cfa}>0$ depending on $\alpha, \oldu{supkernel}, b, c_\gamma, A$, we have
	\begin{align*}
		\E  \sup_{f \in \Fscr_{h,v}(\delta)} \frac{1}{N}\left| \sum_{i=1}^N \gamma_i f({X_i}) \right| & \leq \oldl{cfa} \max\left\lbrace  N^{-1/2} h^{d_1\alpha/4}\delta^{1-\alpha/2}, N^{-2/(2+\alpha)} h^{d_1\alpha/(2+\alpha)} \right\rbrace
	\end{align*}

	When $\delta  \ge  N^{\frac{-1}{2+\alpha}}h^{\frac{d_1\alpha}{2(2+\alpha)}}$, $ \E  \sup_{f \in \Fscr(\delta)} \frac{1}{N} \sum_{i=1}^N{ \gamma}_i f({X_i}) \leq  \oldl{cfa} N^{-1/2} h^{d_1\alpha/4}\delta^{1-\alpha/2}$;
	By Talagrand concentration inequality, for $\ct  \ge 1$,  there exists a constant $\newl\ltxlabel{cfp1}>0$ depending on $\oldu{supkernel}, b, \alpha, \oldu{prop}, A$, such that
	\begin{align*}
		\pr\left( \sup_{f \in \Fscr_{h,v} (\delta)} \frac{1}{N} \left| \sum_{i = 1}^N { \gamma}_i f({X_i}) \right| > 2 \oldl{cfa} \ct  N^{-1/2} h^{d_1\alpha/4}\delta ^{1-\alpha/2} \right)  \leq c\exp\left\{  - \oldl{cfp1} \ct h^{d_1\alpha/2} \delta ^{-\alpha}\right\}.
	\end{align*}
	When $\delta <   N^{\frac{-1}{2+\alpha}}h^{\frac{d_1\alpha}{2(2+\alpha)}} $, $ \E  \sup_{f \in \Fscr_{h,v}(\delta)}| \frac{1}{N} \sum_{i=1}^N { \gamma}_i f({X_i}) |\leq \oldl{cfa} N^{-2/(2+\alpha)} h^{d_1\alpha/(2+\alpha)} $.  Then there exists a constant $\newl\ltxlabel{cfp2}>0$ depending on $\oldu{supkernel}, b, \alpha, \oldu{prop}, A$, such that for $\ct \ge 1$,
	\begin{align*}
		\pr\left( \sup_{f \in \Fscr_{h,v} (\delta)} \frac{1}{N} \left| \sum_{i=1}^N { \gamma}_i f({X_i}) \right| > 2 \oldl{cfa} \ct N^{\frac{-2}{2+\alpha}} h^{\frac{d_1\alpha}{2+\alpha}}   \right)  \leq c\exp\left\{ - \oldl{cfp2} \ct  N^{\frac{\alpha}{2+\alpha}} h^{\frac{d_1\alpha}{2+\alpha}} \right\},
	\end{align*}

	Take $\xi_{N,h}  =  N^{\frac{-1}{2+\alpha}}h^{\frac{d_1\alpha}{2(2+\alpha)}}  $. It is easy to see that $\|f\|^2_2 \leq b^2\oldu{supkernel}^2 h^{d_1}$ for every $f \in \Fscr_{h,v}$.
	We now apply the peeling technique. Take $\ct ' = 2^{2-\alpha/2}\oldl{cfa}\ct $.
	When $\|f\|_2 >\xi_{N,h}$, there exists a constant $\newl \ltxlabel{cfpf_temp}>0$ depending on $\oldu{supkernel}, b, \alpha, \oldu{prop}, A$, such that
	\begin{align*}
		 & \pr\left(\sup_{f \in \Fscr_{h,v}: \xi_{N,h}\leq\|f\|_2\leq \oldu{supkernel}bh^{d_1/2} } \frac{ \frac{1}{N} \left| \sum_{i=1}^N { \gamma}_i f({X_i}) \right| }{\|f\|_2^{1-\alpha/2}} \ge \ct 'N^{-1/2}h^{d_1\alpha/4} \right)                                                                                  \\
		 & \leq \sum_{s=1}^{\left\lceil \log \frac{\xi_{N,h} h^{d_1/2} }{\oldu{supkernel}b} \right\rceil} \pr\left( \sup_{f \in \Fscr_{h,v}: 2^{-s} \oldu{supkernel}b h^{d_1/2}    \leq \|f\|_2  \leq 2^{-s+1}\oldu{supkernel}b h^{d_1/2}   }\frac{1}{N} \left| \sum_{i=1}^N \gamma_i f({X_i}) \right|  \ge \ct ' N^{-1/2}h^{d_1\alpha/4}(2^{-s}\oldu{supkernel}b h^{d_1/2} )^{1-\alpha/2}\right) \\
		 & = \sum_{s=1}^{\left\lceil\log \frac{\xi_{N,h} \sqrt{h}}{\oldu{supkernel}b}\right\rceil} \pr\left( \sup_{f \in \Fscr_{h,v}: 2^{-s} \oldu{supkernel}bh^{1/2}   \leq \|f\|_2  \leq 2^{-s+1}\oldu{supkernel}bh^{1/2}  }\frac{1}{N} \left| \sum_{i=1}^N \gamma_i f({X_i}) \right|  \ge 2\ct \oldl{cfa} N^{-1/2}h^{d_1\alpha/4}(2^{-s+1}\oldu{supkernel}b h^{d_1/2} )^{1-\alpha/2}\right)    \\
		 & \leq \sum_{s=1}^{\infty} c\exp(-\oldl{cfp1} \ct h^{d_1\alpha/2}(2^{-s+1}\oldu{supkernel}bh^{d_1/2})^{-\alpha})                                                                                        \\
		 & = \sum_{s=1}^{\infty} c\exp(-\oldl{cfp1} \ct  (2^{-s+1}\oldu{supkernel} b)^{-\alpha} ) \leq c\exp\left(-{\oldl{cfpf_temp}} \ct' \right).
	\end{align*}
	Therefore,   with probability at least $1-c \exp(-\oldl{cfpf_temp}\ct' )$, we have
	\begin{align}
		\forall f \in \Fscr_{h,v} \qquad \frac{1}{N} \left| \sum_{i=1}^N { \gamma}_i f({X_i}) \right| & \leq \ct' \left\lbrace  N^{-1/2} h^{d_1\alpha/4} \|f\|_2^{1-\alpha/2} + N^{-\frac{2}{2+\alpha}} h^{\frac{d_1\alpha}{2+\alpha}}\right\rbrace,  \label{eqn:Ztern_temp}
	\end{align}
	for any $\ct' \ge 2^{2-\alpha/2} \oldl{cfa}$.

	By H\"older's inequality,
	\begin{align*}
		\|f\|_2^2 = \|f^2\|_1 &
		\leq \|u^2(\cdot)\|_p\left\|K^2\left(\frac{V-\cdot}{h}\right)\right\|_q \leq (b^{2p-2})^{\frac{1}{p}} \|u\|_2^{\frac{2}{p}}h^{\frac{d_1}{q}},
	\end{align*}
	where $p,q\ge 1$ such that $1/p + 1/q =1$.  Plugging this result into \eqref{eqn:Ztern_temp} and taking  $\ct = \ct' \max \{ b^2, 1\}$, we finally get
	\begin{align*}
		\forall f \in \Fscr_{h,v} \qquad \frac{1}{N} \left| \sum_{i=1}^N{\gamma}_i f({X_i}) \right| \leq  \ct \left\{ N^{-\frac{1}{2}}\|u\|_2^{\frac{2-\alpha}{2p}} h^{\frac{d_1\alpha}{4}+ \frac{(2-\alpha)d_1}{4q}} + N^{-2/(2+\alpha)} h^{d_1\alpha/(2+\alpha)} \right\},
	\end{align*}
	with probability at least $1- \exp(-\oldl{cfpf}\ct)$ for $\ct \ge \oldl{cfpf2}$, where $\oldl{cfpf}, \oldl{cfpf2}>0$ are some constants depending on $b$, $\oldu{supkernel}$, $A$,  $\oldu{prop}$ and $\alpha$.
\end{proof}

We then relates $\|u\|_2$ to $\|u\|_N$ in the next lemma.
\begin{lemma}
	\label{u_norm}
	There exist constants $\newl\ltxlabel{cu2}, \newl\ltxlabel{u_order_const} >0$ depending on $b$ and $\alpha$, such that
	for  $\ct \ge \oldl{cu2} $, we have
	with probability  at least $1-\exp(-\oldl{u_order_const} \ct N^{\alpha/(2+\alpha)})$,
	$$\forall u \in \Hscr(1) \qquad \|u\|^2_2 \leq  \ct ( \oldl{u_order_const} N^{-\frac{2}{2+\alpha}}+ \|u\|_N^2).$$
\end{lemma}

\begin{proof}
	Take $r_i$, $i = 1,\dots,n$, as independent rademacher random variables. From the proof of Lemma \ref{entropy}, we know
	$\mathcal{N}(\epsilon,\Hscr(1), \|\cdot\|_\infty)\leq A\epsilon^{-\alpha}$
	for some constant $A>0$.  Therefore, by Theorem 3.12 in \citet{koltchinskii2011oracle}, we have
	$$\E \sup_{u \in \Hscr(1), \|u\|\leq \delta} \left|\frac{1}{N} \sum_{i=1}^{N}r_i u({X_i} )\right| \leq \newl\ltxlabel{u_rade_const} \left(  N^{\frac{-1}{2}} \delta^{1-\frac{\alpha}{2}} + N^{\frac{-1}{1+\alpha/2}}\right),$$
	where $\oldl{u_rade_const}>0$ is a constant depending on $b$ and $\alpha$.

	Next, we will adopt Theorem 3.3 in \cite{bartlett2005local}. Note that
	$$\Var \left\lbrace u^2({X_i})\right\rbrace  \leq  \E \left\lbrace u^4({X_i})\right\rbrace   \leq b^2 \|u\|_2^2 .$$
	Take $\psi(z) := 4\oldl{u_rade_const}b^3  \left(  N^{-1/2} z^{\frac{2-\alpha}{4}} b^{(\alpha-2)/2}+ N^{-1/(1+\alpha/2)}\right)$, $T(u) = b^2\|u\|_2^2$ and  $B = b^2$ in Theorem 3.3 of \cite{bartlett2005local}. It is easy to verify that $\psi(z)$ is non-decreasing   and $\psi(z)/\sqrt{z}$ is non-increasing.  In addition,  we can also verify the condition that for every $z$,
	$$b^2 \E \sup_{u \in \Hscr(1),T(u) \leq z} \left|\frac{1}{N} \sum_{i=1}^{N}r_i u^2({X_i} )\right|  \leq  4b^3 \E \sup_{u \in \Hscr(1),T(u)\leq z} \left|\frac{1}{N} \sum_{i=1}^{N}r_i u({X_i} )\right| \leq   \psi(z).$$
	Then we will find the fixed points $z^*$ of $\psi(z)$ (i.e., the solution of $\psi(z) = z$).
	It can be shown that $z^* = \oldl{u_order_const}N^{-2/(2+\alpha)} $ for some constant $\oldl{u_order_const}$ depending on $\alpha$ and $b$.
	Therefore, Theorem 3.3 in \cite{bartlett2005local} shows that with probability at least $1-\exp\{-\ct Nz^*\}$,
	$$\forall u \in \Hscr(1)  \qquad \|u\|^2_2 \leq  \ct (z^*+ \|u\|_N^2), $$
	with $\ct > \oldl{cu2}$ and a constant $\oldl{cu2}>0$ depending on  $b$ and $\alpha$.

\end{proof}

From Lemmas \ref{Z_term} and \ref{u_norm}, we can see that
for any $\ct_1, \ct_2  \ge \max\{\oldl{cfpf2}, \oldl{cu2}, 1\}$,   with
probability at least $1- \left\{c\exp(-\oldl{cfpf}\ct_1) +\exp(-\oldl{cu2} \ct_2 N^{\alpha/(2+\alpha)})\right\}$,
we have
\begin{align}
	\label{eqn:Z_term1}
	\forall f\in \mathcal{F}_{h,v} \qquad\frac{1}{N} \left| \sum_{i=1}^N { \gamma}_i f({X_i}) \right| & \leq \ct_1 \ct_2 \left\{N^{\frac{-1}{2}}(\|u\|_N)^{\frac{2-\alpha}{2p}} h^{\left( \frac{1}{2} - \frac{2-\alpha}{4p}\right)d_1}
	+ N^{\frac{-2}{2+\alpha}} h^{\frac{d_1\alpha}{2+\alpha}}+
	N^{\frac{-1}{2} -\frac{2-\alpha}{2p(2+\alpha)}}h^{\left( \frac{1}{2} - \frac{2-\alpha}{4p}\right)d_1} \right\}.
\end{align}
Let $s\ge 1$.
Note that $\{u/\|u\|_{\Hscr}: \|u\|_N\le 1\}\subseteq\Hscr(1)$.
Using \eqref{eqn:Z_term1}, we have,
with probability at least $1- \left\{c\exp(-\oldl{cfpf}\ct_1) +\exp(-\oldl{cu2} \ct _2N^{\alpha/(2+\alpha)})\right\}$, uniformly for all $u\in \Hscr$ with $\|u\|_N\le 1$,
\begin{align}
	\frac{1}{N} \left| \sum_{i=1}^N \gamma_i \frac{u({X_i})}{\|u\|_\Hscr} K\left( \frac{{V_i}- v}{h}\right) \right| & \leq \ct_1\ct_2 \left\{N^{\frac{-1}{2}}\left\|\frac{u}{\|u\|_\Hscr}\right\|_N^{\frac{2-\alpha}{2p}} h^{\left( \frac{1}{2} - \frac{2-\alpha}{4p}\right)d_1}
	+ N^{\frac{-2}{2+\alpha}} h^{\frac{d_1\alpha}{2+\alpha}}+
	N^{\frac{-1}{2} -\frac{2-\alpha}{2p(2+\alpha)}}h^{\left( \frac{1}{2} - \frac{2-\alpha}{4p}\right)d_1} \right\}                                                                                                                                                        \nonumber \\
	\frac{1}{N} \left| \sum_{i=1}^N \gamma_i u({X_i})K\left( \frac{{V_i}- v}{h}\right) \right|                      & \leq \ct_1\ct_2 \left\{N^{\frac{-1}{2}}\|u\|_\Hscr^{1-\frac{2-\alpha}{2p}} h^{\left( \frac{1}{2} - \frac{2-\alpha}{4p}\right)d_1}                           
	+ \nu_{N,h}\|u\|_\Hscr \right\} \label{eqn:Zterm_hnorm},
\end{align}
where $\nu_{N,h} :=  N^{-2/(2+\alpha)} h^{d_1\alpha/(2+\alpha)}  + N^{\frac{-1}{2} -\frac{2-\alpha}{2p(2+\alpha)}}h^{\left( \frac{1}{2} - \frac{2-\alpha}{4p}\right)d_1} $, $p\ge 1$.
Next, we define
\begin{align}
	\label{eqn:L_term}
	L (N,h,p,u): = N^{-\frac{1}{2}}\|u\|_\Hscr^{1-\frac{2-\alpha}{2p}} h^{\left( \frac{1}{2} - \frac{2-\alpha}{4p}\right)d_1}  + \nu_{N,h}\|u\|_\Hscr,
\end{align}
for any $N > 1$, $h>0$, $p\ge1$ and $u\in \Hscr$.

Now we are able to bound $S_{N,h}(w^*,u)$ by the following lemma.
\begin{lemma}
	Under Assumption \ref{assum: prop}-\ref{assum: ban},
	\label{Sterm_expbound}
	\begin{align*}
		\sup_{u \in \Hscr_{N}} \frac{S_{N,h}(w^*, u)}{h^{-2d_1}\left\lbrace L^2(N, h, p, u) \right\rbrace }  = \bigOp(1)
	\end{align*}
	where $L$ is defined in \eqref{eqn:L_term},  $p\ge 1$, $h>0$ can depend on $N$.
\end{lemma}

\begin{proof}
	First, take
	$$
		Q(v) := \sup_{u \in \Hscr_{N}} \left|  \frac{\frac{1}{N}\sum_{i=1}^{N} \gamma_iu(X_i) K\left(\frac{V_i-v}{h}\right)}{L(N,h,p,u)}  \right|.
	$$

	Due to \eqref{eqn:Zterm_hnorm},  we can show that for
	any $\ct  \ge \max\{\oldl{cfpf2}, \oldl{cu2}, 1\}$,
	$$Q(v) \leq \ct^2,$$
	with probability at least $1- 2c\exp(-\oldl{cfpf}\ct)$ for large enough $N$.

	Take $\tilde{c} (k)=  \left(\max\{\oldl{cfpf2}, \oldl{cu2}, 1\}\right)^{4k}$.
	From the above upper bound for $Q(v)$,
	we have  for any $v \in [0,1]^{d_1}$ and any integer $k \ge 1$,
	\begin{align*}
		\E \left(Q^2(v) \right)^{k} & = \int_0^\infty \pr\left( Q(v)^{2k} > \ct \right) dt= \int_0^\infty \pr\left( Q(v)> \ct ^{\frac{1}{2k}}\right) dt \\
		                            & \leq  \tilde{ c}(k)  + \int_{ \tilde{c}(k)}^\infty 2c\exp(-\oldl{cfpf}\ct^{\frac{1}{4k}}) dt                      \\
		                            & =  \tilde{ c}(k)  + 4k \int_{ \max\{\oldl{cfpf2}, \oldl{cu2}, 1\} }^\infty 2c\exp(-\oldl{cfpf} t')(t')^{4k-1} dt' \\
		                            & \leq   \tilde{ c}(k)  +  \newl\ltxlabel{const_gamma}k\Gamma(4k),
	\end{align*}
	where $\oldl{const_gamma}>0$ is a constant depending on $\oldl{cfpf}$. Note that for any fixed positive $k$, $\tilde{ c}(k)$ and $k\Gamma(k) $ are bounded.

	From \eqref{eqn:tilde_S},  we have  for $t > 0 $ and positive integer $k$,

	\begin{align}
		 & \pr \left( \sup_{u \in \Hscr_{N}} \frac{\oldl{g_lower_const}^2 h^{2d_1}  \tilde{ S}_{N,h}(w^*,u)}{L^2(N,h,p,u)} \ge t\right)  \leq \pr\left( \left\{  \int_{[0,1]^{d_1}} Q^2(v)  dv   \right\} \ge t  \right) \nonumber \\
		 & \leq  \frac{ \E \left[  \int_{[0,1]^{d_1}} Q^2(v)  dv   \right]^{k} }{t^{k}} \leq  \frac{  \E \left[  \int_{[0,1]^{d_1}} Q^{2k}(v)  dv  \right]}{t^{k}} \qquad \mbox{(by  Jensen's  inequality)}            \nonumber     \\
		 & \leq  \frac{   \int_[0,1]^{d_1} \E Q^{2k}(v)  dv   }{t^{k}}  \leq  \frac{ 2^k (\tilde{ c}(k) + \oldl{const_gamma}k\Gamma(4k))}{t^{k}} \leq \frac{\newl \ltxlabel{Zterm_const}(k)}{t^k},\nonumber
	\end{align}
	where $\oldl{Zterm_const}(k)>0$ is a constant depending on $k$. 
	And then we have
	$$\sup_{u \in \Hscr_{N}} \frac{h^{2d_1}  \tilde{ S}_{N,h}(w^*,u)}{L^2(N,h,p,u)} = \bigOp(1).$$

	From \eqref{eqn:denorm_control} in Lemma \ref{norm_control}, we can see that with probability at least $1-c\exp\left\lbrace  -\oldl{kp2} \ct' Nh\right\rbrace $,  where $\frac{1}{2} \leq \ct' \leq 1$,
	\begin{align}
		\forall \tilde{v} \in [0,1]^{d_1}, \qquad & \left|\frac{1}{Nh} \sum_{i=1}^N K\left( \frac{V_i - \tilde{v}}{h}  \right) - g_h(\tilde{v})\right|  \leq    \ct' \oldl{g_lower_const} \leq \ct'g_h(\tilde{v}) \nonumber \\
		                                          & \frac{1}{Nh} \sum_{i=1}^N K\left( \frac{V_i - \tilde{v}}{h}  \right) - g_h(\tilde{v}) \ge -\ct' g_h(\tilde{v})   \nonumber                                              \\
		                                          & \frac{\frac{1}{Nh} \sum_{i=1}^N K\left( \frac{V_i - \tilde{v}}{h}  \right)}{ g_h(\tilde{v})} \ge 1-\ct'       \label{eqn:d2}                                            \\
		                                          & \frac{ g_h(\tilde{v})}  {\frac{1}{Nh} \sum_{i=1}^N K\left( \frac{V_i - \tilde{v}}{h}  \right)}\leq  \frac{1}{1-\ct'} \nonumber
	\end{align}
	Therefore,
	$$
		\sup_{u \in \Hscr_{N}} \frac{ S_{N,h}(w^*,u)}{h^{-2d_1}L^2(N,h,p,u)}
		\leq  \sup_{u \in \Hscr_{N}} \frac{\tilde{S}_{N,h}(w^*, u)}{h^{-2d_1} L^2(N,h,p,u)} \sup_{\tilde{v}\in[0,1]^{d_1}}\left\lbrace   \frac{ g_h(\tilde{v})}  {\frac{1}{Nh^{d_1}} \sum_{i=1}^N K\left( \frac{V_i - \tilde{v}}{h}  \right)}\right\rbrace ^2
		= \bigOp(1)
	$$

\end{proof}

Next, we control the penalty term $R_{N,h}(w^*)$ through the following lemma.
\begin{lemma}
	\label{V_term}
	Under Assumptions \ref{assum: prop}-\ref{assum: ban},
	$$R_{N,h}(w^*)  = \bigOp (h^{-d_1}).$$ 
\end{lemma}

\begin{proof}

	Take
	\begin{align*}
		\tilde R_{N,h}(w^*) & := \int_{[0,1]^{d_1}}  \frac{1}{g_h(v)^2} \left\lbrace \frac{1}{Nh^{2d_1}}  \sum_{i=1}^N  T_i{w^*_i}^2K^2\left( \frac{ V_i - v}{h}  \right)\right\rbrace   dv.
	\end{align*}
	Notice that $T_i{w^*_i}^2$ is upper bounded by $ \oldu{prop}^2$. By   \eqref{eqn:norm_control} in Lemma \ref{norm_control},
	\begin{align*}
		 & \sup_{\tilde{v}\in [0,1]^{d_1}} \left| \frac{1}{N} \sum_{i=1}^N  T_i{w^*_i}^2K^2\left( \frac{ V_i - \tilde{v}}{h}\right) \right| \leq  \oldu{prop}^2 \sup_{\tilde{v}\in [0,1]^{d_1}} \left| \frac{1}{N} \sum_{i=1}^N  K^2\left( \frac{ V_i - \tilde{v}}{h}\right) \right| \\
		 & =\oldu{prop}^2\knnorm^2 \leq 2 \oldu{prop}^2 \oldl{kernel_exp}\ct h^{d_1},
	\end{align*}
	with probability at least $1-c\exp(-\oldl{kp1} \ct Nh^{d_1})$ for $\ct \ge 1$.
	Therefore,
	\begin{align}
		\label{eqn:vbound}
		\tilde{R}_{N,h}(w^*) & \leq \int_{[0,1]^{d_1}}   \frac{1}{g_h^2(v)}  dv \left\lbrace \frac{1}{h^{2d_1}} 2 \oldu{prop}^2 \oldl{kernel_exp}\ct h^{d_1} \right\rbrace  \leq  \frac{2\oldu{prop}^2\oldl{kernel_exp}^2}{\oldl{g_lower_const}^2}\ct  h^{-d_1},
	\end{align}
	with probability at least $1-c\exp(-\oldl{kp1} \ct Nh^{d_1}/c)$. Combining with the results from \eqref{eqn:d2}, we have
	\begin{align*}
		R_{N,h} (w^*)  \leq \tilde{R}_{N,h}(w^*)\left\lbrace  \sup_{\tilde{v} \in [0,1]^{d_1}} \frac{ g_h(\tilde{v})}  {\frac{1}{Nh} \sum_{i=1}^N K\left( \frac{V_i - \tilde{v}}{h}  \right)} \right\rbrace ^2
		= \bigOp(h^{-d_1})
	\end{align*}

\end{proof}

Now, we are ready to prove Theorem \ref{thm: prep}.
\begin{proof}[Proof of Theorem \ref{thm: prep}]
	Take $u^* = \argmax_{u \in \Hscr_N } \left\lbrace  S_{N,h} (w^*, u) - \lambda_1 \|u\|_\Hscr^2 \right\rbrace$. Its existence is shown in Section \ref{sec:comp_supp}.

	Due to \eqref{eqn:obj_all}, we have the following basic inequality:
	\begin{align}
		\label{eqn:basic_inequality_1}
		S_{N,h}(\hat{w}, m_1) + \lambda_1\|u^*\|^2_\Hscr \|m_1\|^2_N + \lambda_2 R_{N,h} (\hat w)\|m_1\|_N^2 \leq S_{N,h}(w^*, u^*)\|m_1\|_N^2
		+ \lambda_1 \|m_1\|^2_\Hscr + \lambda_2R_{N,h}(w^*)\|m_1\|_N^2
	\end{align}
	From Lemmas \ref{Sterm_expbound} and \ref{V_term}, we have
	$R_{N,h}(w^*)  = \bigOp (h^{-d_1})$ and
	\begin{equation*}
		S_{N,h}(w^*, u^*) = \bigOp \left( N^{-1}\|u^*\|_\Hscr^{2-\frac{2-\alpha}{p}} h^{\left( -1-\frac{2-\alpha}{2p}\right)d_1} + \nu^2_{N,h} h^{-2d_1} \|u^*\|^2_\Hscr  \right)
	\end{equation*}
	for all $p\ge 1$.

	We now compare different scenarios of \eqref{eqn:basic_inequality_1}.

	Case 1:  Suppose that $S_{N,h}(w^*, u^*)\|m\|_N^2$ is the largest in the right-hand side of \eqref{eqn:basic_inequality_1}.

	If $\|m\|_N \neq  0$, we have
	$\lambda_1\|u^*\|^2_\Hscr \leq \bigOp \left( N^{-1}\|u^*\|_\Hscr^{2-(2-\alpha)/p} h^{-1 -\frac{2-\alpha}{2p}} \right)  + \bigOp \left( \nu^2_{N,h} h^{-2} \|u^*\|^2_\Hscr \right)$.
	By Assumptions \ref{assum: ratio} and \ref{assum: ban}, we can see that
	\begin{align*}
		\nu_{N,h}^2h^{-2} & = N^{-\frac{4}{2+\alpha}} h^{\left(\frac{2\alpha}{2+\alpha} -2 \right)d_1}  + N^{-1 -\frac{2-\alpha}{p(2+\alpha)}}h^{\left(1 - \frac{2-\alpha}{2p}-2 \right)d_1} \\
		                  & =(N^{-1}h^{-d_1})^{\frac{4}{2+\alpha}} + (N^{-1} h^{-d_1})(h^{-\frac{d_1}{2}}N^{-\frac{1}{2+\alpha}})^{\frac{2-\alpha}{p}}                                       \\
		                  & = \bigO(N^{-1}h^{-d_1}) = \smallO(\lambda_1)
	\end{align*}
	Therefore we only need to consider $\lambda_1\|u^*\|^2_\Hscr \leq \bigOp \left( N^{-1}\|u^*\|_\Hscr^{2-(2-\alpha)/p} h^{\left(-1 + \frac{\alpha-2}{2p} \right)d_1} \right)$ .
	Then we have
	\begin{align*}
		\|u^*\|_\Hscr & \leq \lambda_1^{-\frac{p}{(2-\alpha)}} \bigOp \left(  N^{-\frac{p}{(2-\alpha)}} h^{\left( {-\frac{p}{(2-\alpha)}} - \frac{1}{2}\right)d_1}
		\right),
	\end{align*}
	and
	$$S_{N,h} (\hat w , m) \leq \lambda_1^{\frac{-2p + (2-\alpha)}{(2-\alpha)}} \bigOp \left(N^{\frac{-2p}{(2-\alpha)}} h^{\left(\frac{-2p}{(2-\alpha)}-1 \right)d_1} 
		\right)\|m\|_N^2
		.$$
	If $\|m\|_N = 0$
	we have $S_{N,h}(\hat w, m) = 0 \leq  \lambda_1^{\frac{-2p + (2-\alpha)}{(2-\alpha)}} \bigOp \left(N^{\frac{-2p}{(2-\alpha)}} h^{\left(\frac{-2p}{(2-\alpha)}-1 \right)d_1}\right) \|m\|_N^2$.

	Case 2: Suppose that $\lambda_1 \|m\|^2_\Hscr$ is the largest in right-hand side of \eqref{eqn:basic_inequality_1}. Then we have $S_{N,h}(\hat w, m) \leq 3\lambda_1 \|m\|_\Hscr^2 = \bigOp(\lambda_1) \|m\|_\Hscr^2$.

	Case 3: Suppose that $\lambda_2 R_{N,h}(w^*)$ is the largest in right-hand side of \eqref{eqn:basic_inequality_1}.  Then we have $S_{N,h}(\hat w, m) \leq 3\lambda_2 \bigOp(h^{-d_1}) \|m\|_N^2 = \bigOp ( \lambda_2 h^{-d_1}  )\|m\|_N^2$.

Combining these cases, we have
	\begin{align}
		S_{N,h}(\hat w, m_1) = \max\left\{ \min\left\{\lambda_1^{\frac{-2p + (2-\alpha)}{(2-\alpha)}} \bigOp \left(N^{\frac{-2p}{(2-\alpha)}} h^{\left(\frac{-2p}{(2-\alpha)}-1 \right)d_1}\right)\|m_1\|_N^2: p \ge 1\right\}, \right. \nonumber\\
		\left.\bigOp(\lambda_1) \|m_1\|_\Hscr^2, \bigOp( \lambda_2 h^{-d_1})\|m_1\|_N^2 \right\}.\label{Sterm_orders}
	\end{align}

	Next, we compare the first two components of \eqref{Sterm_orders}.
	We can see that as long as $$\frac{2p}{2-\alpha}\log(\lambda_1^{-1}N^{-1}h^{-d_1}) \leq \log h^{d_1},$$ the second component is dominant. Note that $\log h^{d_1} < 0$ as $h \rightarrow 0$.  Because of the condition that $\lambda_1^{-1} = \smallO(Nh^{d_1})$, the inequality is valid as long as $p \ge \frac{2-\alpha}{2}\frac{\log h^{d_1}}{\log (\lambda_1^{-1}N^{-1}h^{-d_1})}$. So we can pick any $p \ge \max\{1, \frac{2-\alpha}{2}\frac{\log h^{d_1}}{\log (\lambda_1^{-1}N^{-1}h^{-d_1})}\}$ to have the best order $\bigOp(\lambda_1)(\|m\|_\Hscr^2 + \|m\|_N^2)$.

	Then, we compare the first and the third components of \eqref{Sterm_orders}. Similar to the previous analysis, as long as
	$$ \frac{2p}{2-\alpha}\log(\lambda_1^{-1}N^{-1}h^{-d_1}) \leq \log (\lambda_2 \lambda_1^{-1}),$$
	the third component is dominant.  Due to the condition that $\lambda_1^{-1} = \smallO(Nh^{d_1})$, the inequality is valid if $p \ge \frac{2-\alpha}{2}\frac{\log \lambda_2\lambda_1^{-1}}{\log (\lambda_1^{-1}N^{-1}h^{-d_1})}$. So we can pick any $p \ge \max\{1, \frac{2-\alpha}{2}\frac{\log \lambda_2\lambda_1^{-1}}{\log (\lambda_1^{-1}N^{-1}h^{-d_1})}\}$ to have the best order $\bigOp(\lambda_2h^{-d_1})\|m\|_N^2$.

	Finally, we conclude that $$S_{N,h}(\hat w, m_1) = \bigOp(\lambda_1 \|m_1\|_{N}^2 + \lambda_1 \|m_1\|_\Hscr^2 + \lambda_2h^{-d_1}\|m_1\|_N^2).$$
	Moreover, further suppose that $\lambda_2^{-1} = \bigO(\lambda_1^{-1} h^{-d_1}) $.   From \eqref{eqn:basic_inequality_1}, by replacing $m$ with a constant function and applying the similar analysis as above, we can conclude that $R_{N,h}(\hat w) = \bigOp(h^{-d_1})$.

\end{proof}

\subsection{Proof of Theorem \ref{thm: decomp}}
\label{sec:proof_decomp}
\begin{proof}
First,
	\begin{align}
		     & \left\|\frac{1}{N}\sum_{i=1}^{N} T_i\hat w_iY_i K_h\left(  V_i,\cdot \right) -\E\left\lbrace Y(1) \mid  V= \cdot \right\rbrace  \right\|_2                                            \nonumber \\
		\leq & \left\| \frac{1}{N}\sum_{i=1}^{N} (T_i\hat w_i-1) K_h\left(  V_i,\cdot \right)m(X_i) \right\|_2 + \left\|\frac{1}{N} \sum_{i=1}^N T_i\hat w_i K_h(V_i,\cdot )\varepsilon_i  \right\|_2 \label{eqn:eq2}\\
		+    & \left\| \frac{1}{N}\sum_{i=1}^{N}m(X_i)K_h(V_i,\cdot )  - \E \left\lbrace Y(1) \mid  V= \cdot \right\rbrace  \right\|_2. \label{eqn:eq3}
	\end{align}
	Since $\|m_1\|_2 \leq b \|m_1\|_\Hscr < \infty$ and $\|m_1\|_N = \|m_1\|_2 + \smallOp(1)$, we have
	\begin{align*}
		 & \left\| \frac{1}{N}\sum_{i=1}^{N} (T_i\hat w_i-1) K_h\left(  V_i,\cdot \right)m_1(X_i) \right\|_2 = S^{1/2}_{N,h}(\hat{w}, m_1) = \bigOp(\lambda_1^{1/2}\|m_1\|_N + \lambda_1^{1/2}\|m_1\|_\Hscr + \lambda_2^{1/2}h^{-d_1/2}\|m_1\|_N)  \\
		 & =\bigOp(\lambda_1^{1/2}\|m_1\|_\Hscr + \lambda_2^{1/2}h^{-d_1/2}\|m_1\|_2) + \smallOp(\lambda_1^{1/2} + \lambda_2^{1/2}h^{-d_1/2})
	\end{align*}
	due to Theorem \ref{thm: prep} and the conditions of $\lambda_1$ and $\lambda_2$.

	For the second term in \eqref{eqn:eq2}, we have $\E (\varepsilon_i \mid  T_i, \hat{w}_i, X_i,  i = 1,\dots, N) = 0$. Then
	\begin{align*}
		     & \E \left\{  \left\|\frac{1}{N} \sum_{i=1}^N T_i\hat w_i K_h(V_i,\cdot )\varepsilon_i  \right\|_2^2  \mid T_i, \hat{w}_i, X_i,  i = 1,\dots, N \right\}      \\
		=    & \int_0 ^1 \E \left\{ \left[ \frac{1}{N} \sum_{i=1}^N T_i\hat w_i K_h(V_i,v )\varepsilon_i \right] ^2  \mid T_i, \hat{w}_i, X_i,  i = 1,\dots, N \right\} dv \\
		=    & \int_0^1  \frac{1}{N^2} \sum_{i=1}^N \E \left\{  T_i\hat{w}^2_iK^2_h(V_i,v)\epsilon^2_i   \mid T_i, \hat{w}_i, X_i,  i = 1,\dots, N \right\} dv             \\
		\leq & \frac{\sigma_0^2}{N}  \int_{ 0}^1 \frac{1}{N} \sum_{i=1}^N  T_i\hat{w}^2_iK^2_h(V_i,v) dv = \frac{\sigma_0^2}{N} R_{N,h}(\hat w).
	\end{align*}

	Therefore,
	\begin{align*}
		\E \left\{\frac{\left\|\frac{1}{N} \sum_{i=1}^N T_i\hat w_i K_h(V_i,\cdot )\varepsilon_i  \right\|^2_2 }{R_{N,h}(\hat w)}\mid T_i, X_i, i = 1,\dots,N \right\} & \leq  \frac{\sigma_0^2}{N}       \\
		\E \left\{\frac{\left\|\frac{1}{N} \sum_{i=1}^N T_i\hat w_i K_h(V_i,\cdot )\varepsilon_i  \right\|^2_2 }{R_{N,h}(\hat w)}\right\}                              & \leq  \frac{\sigma_0^2}{N}       \\
		\frac{\left\|\frac{1}{N} \sum_{i=1}^N T_i\hat w_i K_h(V_i,\cdot )\varepsilon_i  \right\|_2^2}{R_{N,h}(\hat w)}                                                 & = \sigma_0^2\bigOp (\frac{1}{N})
	\end{align*}
	From the condition of $\lambda_2$, and  the  result from Theorem \ref{thm: prep} that $R_{N,h}(\hat w)= \bigOp(h^{-d_1})$, we have
	$$\left\|\frac{1}{N} \sum_{i=1}^N T_i\hat w_i K_h(V_i,\cdot )\varepsilon_i  \right\|_2 = \bigOp(N^{-1/2}h^{-d_1/2}).$$

	As for \eqref{eqn:eq3}, it has a form of a typical Nadaraya–Watson estimator. 
	By Theorem 5.44 in \cite{wasserman2006all} we have
$$\E \left\| \frac{1}{N}\sum_{i=1}^{N}m(X_i)K_h(V_i-\cdot)  - \E \left\lbrace Y(1) | V= \cdot\right\rbrace \right\|_2^2 = \bigO(N^{-1}h^{-d_1}).$$
Therefore,
$$\left\| \frac{1}{N}\sum_{i=1}^{N}m(X_i)K_h(V_i-\cdot)  - \E \left\lbrace Y(1) | V= \cdot\right\rbrace \right\|_2^2 = \bigOp(N^{-1}h^{-d_1}).$$
	Overall, conclusion follows.

\end{proof}

\subsection{Proof outline of Theorem \ref{thm:sup}}
\label{sec:proof_sup}
To obtain the rate, the entropy bound in Lemma \ref{entropy} needs to be modified to the bigger function class $\Fscr_h : = \{f: f(\check{x}) = u(\check{x})K(\frac{\check{v}-v}{h}), u \in \{u\in \Hscr: \|u\|_\Hscr\leq 1\}, v \in [0,1]^{d_1}\}$. This can be done by combining the entropy bound for $\{u\in \Hscr: \|u\|_\Hscr\leq 1\}$ and Assumption \ref{assum: kernel_2}(b). One can show that
\begin{equation*}
	H(\delta, \Fscr_{h}, \|\cdot\|_N)
	\begin{cases}
		= 0                                                                 & \mbox{if $\delta>2b\knnorm$} \\
		\le A\knnorm^{\alpha}\delta^{-\alpha} + \log( A_1\epsilon^{-\nu_1}) & \mbox{otherwise}
	\end{cases}.
\end{equation*}
Then by adopting this entropy bound, the results in Lemma \ref{Z_term} will be modified to
$$
	\forall f \in \Fscr_{h} \qquad \frac{1}{N} \left| \sum_{i=1}^N {\gamma}_i f({X_i}) \right| \leq \ct \left\{ N^{-\frac{1}{2}}\|u\|_2^{\frac{2-\alpha}{2p}} h^{d_1\left(\frac{1}{2} - \frac{2-\alpha}{4p} \right)}\left(\log \frac{1}{h}\right)^{1/2} + N^{-\frac{2}{2+\alpha}} h^{\frac{d_1\alpha}{2+\alpha}}\log \frac{1}{h}\right\},
$$
for any $\ct \ge c_1$, and $p \ge 1$ with probability at least $1- c\exp\left(-{\oldl{cfpf}}\ct  \right)$.
Then the remaining argument is similar to those in the proof of Theorems \ref{thm: prep} and \ref{thm: decomp}.

\subsection{Proof of Theorem \ref{thm: aug}}
\label{sec:proof_aug}

\begin{proof}
	Following the same proof structure of Theorem \ref{thm: prep}, by replacing $m$ with a constant function $z$ of value 1, we have
	\begin{align*}
		S_{N,h}(\hat w, z) + \lambda_1 \|u^*\|_\Hscr^2 	+ \lambda_2 R_{N,h}(\hat w) \leq S_{N,h}(w^*, u^*) + \lambda_1 \|z\|_\Hscr^2 + \lambda_2 R_{N,h}(w^*).
	\end{align*}
	By the condition of $\lambda_1$ such that $\lambda_1^{-1} = \bigOp(N^{-1}h^{-d_1})$, we have
	$R_{N,h}(\hat w) = \bigOp(\lambda_2^{-1} \lambda_1 + h^{-d_1}) $. Since $\lambda_2^{-1}\lambda_1 = \bigO(h^{-d_1})$,
	\begin{align*}
		R_{N,h}(\hat w) = \bigOp(h^{-d_1}).
	\end{align*}

	Again, following the same proof structure of Theorem \ref{thm: prep}, by replacing $m$ with $\hat e$,
	we have
	\begin{align*}
		S_{N,h}(\hat w, \hat e) + \lambda_1 \|u^*\|_\Hscr^2 \|\hat e\|_N^2 	+ \lambda_2 R_{N,h}(\hat w)\|\hat e\|_N^2 \leq S_{N,h}(w^*, u^*)\|\hat e\|_N^2 + \lambda_1 \|\hat e\|_\Hscr^2 + \lambda_2 R_{N,h}(w^*) \|\hat e\|_N^2.
	\end{align*}
	By the condition of $\lambda_1$ such  that $\lambda_1^{-1} = \bigOp(N^{-1}h^{-d_1})$, we can obtain
	\begin{align*}
		S_{N,h}(\hat w, e) = \bigOp(\lambda_1\|e\|_N^2 + \lambda_1 \|e\|_{\Hscr}^2 + \lambda_2 h^{-d_1} \|e\|_N^2).
	\end{align*}
	Therefore,
	\begin{align*}
		     & \left\|\frac{1}{N} \sum_{i=1}^N  \tilde K_h (V_i, \cdot) \hat m(X_i) + \frac{1}{N} \sum_{i=1}^N T_i \hat w_i \tilde K_h (V_i, \cdot) \{ Y_i - \hat m(X_i)\} -\E\left\lbrace Y(1) | V= \cdot\right\rbrace  \right\|_2 \\
		\leq & \left\| \frac{1}{N}\sum_{i=1}^{N} (T_i\hat w_i-1) K_h\left(  V_i,\cdot \right)e(X_i) \right\|_2 + \left\|\frac{1}{N} \sum_{i=1}^N T_i\hat w_i K_h(V_i,\cdot )\varepsilon_i  \right\|_2                   \\
    		&\quad + \left\| \frac{1}{N}\sum_{i=1}^{N}m(X_i)K_h(V_i,\cdot )  - \E \left\lbrace Y(1) \mid  V= \cdot \right\rbrace  \right\|_2                                                                                  \\
		\leq & \{S_{N,h}(\hat w, e)\}^{1/2} + \bigOp(N^{-1/2})R^{1/2}_{N,h}(\hat w) + \bigOp(N^{-1/2}h^{-d_1/2})                                                                                        \\
		\leq & \bigOp(\lambda_1^{1/2}\|e\|_N + \lambda_1^{1/2} \|e\|_{\Hscr} + \lambda_2^{1/2} h^{-d_1/2} \|e\|_N) + \bigOp(N^{-1/2}h^{-d_1/2})                              \\
		\leq & \bigOp(N^{-1/2}h^{-d_1/2} + \lambda_1^{1/2}\|e\|_N + \lambda_1^{1/2} \|e\|_{\Hscr} + \lambda_2^{1/2} h^{-d_1/2} \|e\|_N).
	\end{align*}

\end{proof}

\bibliographystyle{chicago}

\bibliography{refer2.bib}

\makeatletter\@input{xx.tex}\makeatother